\documentclass[a4paper,onecolumn, accepted = 2026-06-26, 11pt]{quantumarticle}
\pdfoutput=1
\usepackage[utf8]{inputenc}
\usepackage[english]{babel}
\usepackage[T1]{fontenc}
\usepackage{amsmath, amsfonts, amssymb, amsthm}
\usepackage{hyperref}
\usepackage{braket}
\usepackage{tikz}
\usepackage{booktabs}
\usepackage{array}
\usepackage{bm}
\usepackage{caption}
\usepackage{aligned-overset}
\usepackage{subcaption}
\usepackage[dvipsnames]{xcolor}

\theoremstyle{plain}
\newtheorem{theorem}{Theorem}
\newtheorem{lemma}[theorem]{Lemma}

\newtheorem{proposition}[theorem]{Proposition}
\newtheorem{fact}[theorem]{Fact}

\theoremstyle{definition}
\newtheorem{definition}[theorem]{Definition}
\newtheorem{example}[theorem]{Example}
\newtheorem{problem}[theorem]{Problem}
\newtheorem{remark}[theorem]{Remark}
\newtheorem{observation}[theorem]{Observation}

\usepackage{xspace}

\newcommand{\ie}{i.e.}
\newcommand{\eg}{e.g.}

\newcommand{\dd}{\,\mathrm{d}}
\newcommand{\ee}{\mathop{\mathbb{E}}}

\newcommand{\etc}{etc.}

\newcommand{\Span}{\textnormal{span}}

\newcommand{\MO}{\mathsf{MO}}

\newcommand{\class}[1]{\mathsf{#1}}

\newcommand{\Tr}{\textnormal{Tr}}

\newcommand{\poly}{\mathrm{poly}}
\newcommand{\polylog}{\mathrm{polylog}}

\newcommand{\Dens}{\mathfrak{D}}

\newcommand{\U}{\mathsf{U}}

\usepackage{mathrsfs}

\newcommand{\cptp}{\mathsf{CPTP}}

\newcommand{\partition}[2]{\vdash_{#1}{#2}}

\newcommand{\haar}{\textnormal{Haar}}

\newcommand{\hil}{\mathcal{H}}

\newcommand{\ml}[1]{\mathcal{#1}}
\newcommand{\mb}[1]{\mathbb{#1}}

\newcommand{\bo}{\mathcal{B}}
\renewcommand{\Ket}[1]{| #1 \rangle \! \rangle}

\allowdisplaybreaks[1]

\begin{document}

\title{Resource quantification for programming low-depth quantum circuits}

\author{Entong He}
\orcid{0009-0001-3911-0966}
\email{ethe@cs.hku.hk}
\author{Yuxiang Yang}
\orcid{0000-0002-0531-8929}
\affiliation{
QICI Quantum Information and Computation Initiative, School of Computing and Data Science,\\ The University of Hong Kong, Pokfulam Road, Hong Kong SAR, China
}
\email{yuxiang@cs.hku.hk}


\maketitle

\begin{abstract}
  Noisy intermediate-scale quantum (NISQ) devices pave the way for implementing quantum algorithms that offer quantum advantages over their classical counterparts. Due to the intrinsic noise and decoherence in the physical system, NISQ machines are naturally modeled as large-scale, low-depth quantum circuits. 
  In practice, executing such circuits requires sending program states that encode the relevant instructions to a programmable quantum computer, typically through a cloud service. Existing programming approaches designed for generic unitary transformations are computationally inefficient in the low-depth setting, and therefore remain unsatisfactory. As such, to realize NISQ algorithms, it is crucial to find an efficient way to program low-depth circuits as the number of qubits $N$ increases. Here, we investigate the circuit complexity and the size of quantum memory, known as the program cost, required to program low-depth brickwork circuits. 
  We establish a tight worst-case program cost of $\Theta(N\polylog N)$ for universally programming low-depth brickwork circuits in the large-$N$ regime. Moreover, we analyze the trade-off between the cost of describing the layout of local gates and the cost of programming them to implement the target unitaries via the light-cone argument. Our findings suggest that faithful gate-wise programming is essentially optimal in the low-depth regime.
\end{abstract}

\section{Motivation}
\label{sec:motivation}
The world is awaiting the arrival of the Noisy Intermediate-scale Quantum (NISQ) technology era in the near future \cite{Preskill2018quantumcomputingin}. With qubit counts on the order of tens to hundreds, quantum computers can run quantum algorithms that offer advantages in time and quantum-memory complexity over their classical counterparts on the same tasks \cite{GroversAlgorithm1996, Shor1997, HHL}. Despite this optimistic vision, the computational power of quantum computers is intrinsically limited by noise inherent in gate operations and by the limited coherence times of quantum particles interacting with the external environment \cite{NISQAlgorithms2022}. To suppress their interference with the computational results, an intuitive approach is to compress the depth of the quantum circuits \cite{YunchaoLiuThesis2024}. These low-depth designs limit the number of sequential operations performed during a computation but can still be provably more powerful than their classical counterparts \cite{Bravyi2018, QNC0NotEqualAC0}. 

\par 
Given the prohibitive resource requirements and costs of local quantum hardware in the NISQ era, clients may delegate computations \cite{DelegatedQuantumComputation2015} to cloud servers by transmitting program states that encode target unitary operations. Rather than designing a dedicated circuit for each computational task, the cloud computer uses a fixed circuit architecture capable of processing \textit{any} program state; we refer to this architecture as its \textit{processor}. The processor's ability to implement any target unitary from its corresponding program state is referred to as \textit{universal programmability}. Programming unitaries differs from synthesizing unitaries \cite{OneQueryUnitarySynthesis} in that all information about the target unitary is supplied in quantum rather than classical form, as a quantum state that we conventionally refer to as its \textit{program state}.
An essential component for implementing NISQ algorithms is a quantum processor capable of universally programming low-depth quantum operations. Although the \textit{No-Programming Theorem} has ruled out the possibility of exact universal programming, as an infinite-size quantum memory and quantum circuit are necessary to store and retrieve the program, its approximate analogue is available \cite{Nielsen1997}. Numerous efforts have been made \cite{ProbabilisticUniversalProcessor2002, MBQComp2007, AsymptoticTeleporation2008, Kubicki2019, YuxiangOptimal2020, Gschwendtner2021programmabilityof, Muguruza2024portbasedstate, yoshida2025quantumadvantagestorageretrieval} 
to circumvent this no-go theorem by allowing nonzero programming error in exchange for reduced quantum-memory complexity; we refer to this complexity as the \textit{program cost}. 
Optimal trade-offs have been established for unitaries \cite{YuxiangOptimal2020}, channels \cite{Gschwendtner2021programmabilityof}, and isometries \cite{yoshida2025quantumadvantagestorageretrieval}, primarily within two frameworks: measure-and-operate (MO) \cite{BisioOptimalLearning2010, YuxiangOptimal2020, Gschwendtner2021programmabilityof}, and port-based teleportation (PbT) \cite{AsymptoticTeleporation2008, Muguruza2024portbasedstate, yoshida2025quantumadvantagestorageretrieval}. Nevertheless, all these works, including \cite{YuxiangOptimal2020}, assume constant system dimensions and derive the bounds for program cost in the arbitrarily small error regime. Meanwhile, NISQ algorithms might run on many qubits to maintain their potential quantum advantage \cite{Bravyi2018}, thereby creating a ``wide'' circuit architecture. Rather than taking the programming error to be infinitesimal, we allow the programming error to vanish as the system dimension grows. Under these assumptions, the number of qubits controls both the asymptotic error and the circuit architecture, making it natural to ask whether finer-grained lower and upper bounds can be obtained for programming low-depth quantum circuits.
In addition, these works focus on an information-theoretic perspective and do not consider the efficiency of implementing the protocols on practical quantum devices, in terms of both gate and program cost.

\par

This article aims to identify the resource requirements for programming $\class{QNC}$ circuits, \ie, quantum circuits with polylogarithmic depth \cite{Arora2009} and bounded-fan-in gates. This gives rise to a fundamental yet representative quantum circuit architecture -- the brickwork circuit \cite{Haferkamp2022, Belkinetal2024, Zhao_2024, HuangLearnShallowCircuit2024, Yang_2025}, which consists of local unitary gates with restricted connectivity. This architecture closely reflects the physical constraints of NISQ devices. Many near-term quantum algorithms are developed with $\mathsf{QNC}$ brickwork circuits. Representative examples include the Variational Quantum Algorithms \cite{Cerezo_2021}, random unitaries \cite{Brand_o_2016, schuster2025randomunitariesextremelylow}, and quantum algorithms that efficiently solve the 2D Hidden Linear Function problem \cite{Bravyi2018}. We quantify the overall circuit complexity of preparing and executing the program using the best-known MO universal unitary programming scheme \cite{YuxiangOptimal2020}. Building on the constructive techniques employed in \cite{YuxiangOptimal2020, Gschwendtner2021programmabilityof}, we provide a tight lower bound for the worst-case program cost requirement, as a refinement of the previous results \cite{Kubicki2019, YuxiangOptimal2020, Gschwendtner2021programmabilityof} in the low-depth regime. A straightforward counting argument of the covering net for all $\class{QNC}$ brickwork circuits asymptotically matches this lower bound. Simple as it is, this argument provides a complete characterization of the quantum-storage complexity of programming low-depth brickwork circuits.
\par Our main results are summarized as follows:

\begin{theorem}[Circuit complexity of optimal universal programming, informal]
\label{thm:GateComplexityInformal}
    The optimal universal programming of $\U(d)$ where $d = 2^N$ with diamond norm error $\epsilon$ can be implemented with $\widetilde{\mathcal{O}}\left( \poly(d, 1 / \sqrt{\epsilon}) \right)$ quantum gates and $\mathcal{O}\left(d^2 \log(d^2 / \epsilon)\right)$ ancillae.  
    Furthermore, programming low-depth brickwork circuits by programming each local gate optimally can be implemented with $\widetilde{\mathcal{O}}\left(\poly(N,  1 / \sqrt{\epsilon}) \right)$ quantum gates and $\widetilde{\mathcal{O}}\left(N \log (1 / \epsilon) \right)$ ancillae.
     
\end{theorem}

We refer to Section~\ref{subsec:CircuitComplexityForUnitaryProgramming} for further context. This argument says that, in terms of \textit{universal} programming of unitary operations, the circuit complexity required to optimally synthesize a generic unitary from its corresponding program state is expected to be exponentially high, even approximately. The hardness originates from the fact that a general unitary matrix is characterized by up to $\Theta(d^2)$ free parameters, so is its program state. Notably, the same scaling behavior appears in the context of quantum gate learning and metrology. A related result is stated in \cite[Section 4.5.4]{Nielsen2012}; in the programming setting, however, the execution circuit is oblivious to the target unitary rather than tailored to it.

\begin{theorem}[Program cost bounds for programming low-depth circuits, informal version of \autoref{thm:DimensionLowerbound} and \autoref{thm:DimensionUpperbound}]
\label{thm:ProgramCostBoundInformal}
    Programming a low-depth quantum circuit on $N$ qubits to diamond-norm error $\epsilon \sim 1 / \polylog N < \frac{1}{32}$ requires a quantum processor with program cost $c_P = \Omega \left( N \polylog N \right)$ in the worst case. Moreover, $c_P = \mathcal{O} \left( N \polylog N \right)$ if we restrict to brickwork circuits.
\end{theorem} 
The arguments presented in \autoref{thm:GateComplexityInformal} and \autoref{thm:ProgramCostBoundInformal} show that the construction in \cite{YuxiangOptimal2020} is far from optimal in terms of both circuit complexity and the cost-error trade-off for programming low-depth circuits.

Notably, the scaling of the cost stated in \autoref{thm:ProgramCostBoundInformal} is tight. Although the error is bounded in our setting, plainly substituting it into the lower and upper bounds for programming general unitaries  
\cite{ProbabilisticUniversalProcessor2002, MBQComp2007, AsymptoticTeleporation2008, Kubicki2019, YuxiangOptimal2020}
 does not give us the desired scaling. Our result is obtained via an information-theoretic approach and the counting argument based on the structure of brickwork circuits.  A compatible scaling for circuit complexity\footnote{
 In their definition, the circuit complexity provides a lower bound for the number of small local gates required to synthesize a global unitary.
 } of $N$-qubit unitaries is reported in \cite[Theorem 1]{Haferkamp2022}. Compared with the cardinality of the $\epsilon$-net of the topological group $\U(d)$, one can conclude that within $\U(d)$, the low-depth brickwork circuit unitaries are sparsely distributed.

 Besides the complexity of individual local unitaries, the program cost is also related to the quantity of unitary gates and the circuit architecture \cite[Figure 2]{Haferkamp2022}. These factors are inversely correlated: For a circuit with fixed geometry, larger and more complicated local gates often result in a simpler layout involving fewer gates. When the local gates have fixed dimensions, we can still transit to the larger unitary case using the widely adopted light-cone argument \cite{Haferkamp2022, HuangLearnShallowCircuit2024, Nadimpalli_2024, Yang_2025}. However, when we take into account the overall program cost, in major cases, the optimal programming scenario is to directly program the primitive small unitaries faithfully.

\begin{fact}[Summary of Section~\ref{sec:TradeOff}, informal]
    There exists an approach to trade the cost of programming local unitary gates in a brickwork quantum circuit off for lower circuit architecture complexity. 
    In terms of overall program cost, however, it provides no reduction in major cases.
\end{fact}

\section{Preliminaries}

We use the following notion in the upcoming context: Let $d = 2^N$ where $N \in \mathbb{N}$, $\mathcal{H} \cong \mathbb{C}^d$ stands for a $d$-dimensional Hilbert space, $\bo(\mathcal{H})$ stands for the set of bounded linear operators on $\mathcal{H}$, and $\U(d)$ stands for the set of $d \times d$ unitary matrices. The set of density matrices on the space $\mathcal{H}$ is denoted as $\Dens(\mathcal{H})$. A quantum operation connecting two Hilbert spaces $\mathcal{H}_A$ and $\mathcal{H}_B$ can be formalized as a quantum channel, which is a completely positive and trace-preserving (CPTP) map. We denote it by $\cptp(\mathcal{H}_A, \mathcal{H}_B)$, or $\cptp(\mathcal{H})$ for short if $\hil_A = \hil_B = \hil$. For a pure state $\ket{\psi}$, we use $\psi$ as a shorthand for its density matrix $\ket{\psi} \bra{\psi}$. For any matrix $A$, we use the double-ket notation $\Ket{A} = \sum_{m, n} \bra{n} A \ket{m} \ket{n} \ket{m}$, and use $\|A\|$ to indicate its operator norm. For any matrix $A, B$, $A \preceq B$ indicates $B - A$ is positive-semidefinite. For any space $\mathcal{H}$, we use $\ket{\Phi_{\mathcal{H}}^+} = \frac{1}{\sqrt{\dim \mathcal{H}}} \Ket{I_{\mathcal{H}}}$ to denote the maximally entangled state on $\mathcal{H}^{\otimes 2}$. For a unitary operator $U$, we use the calligraphic font $\mathcal{U}(\cdot)$ to represent the channel $U (\cdot) U^{\dagger}$. 

We will use the conventional notations (big-$O$, big-$\Omega$ and big-$\Theta$) for the asymptotic behavior of functions \cite{Arora2009}. Besides, we will write $f = o(g)$ and $f = \omega(g)$ if $f(x)/g(x)$ and $g(x)/f(x)$ is vanishing with large $x$. For conciseness, we write $f \sim g$ for $f = \Theta(g)$ and $f \lesssim g$ for $f = \mathcal{O}(g)$. The notation $\widetilde{\mathcal{O}}$, as a variant of the big-$O$, ignores the logarithmic factors. We say a general quantum operation is implemented $\epsilon$-approximately if the circuit produced is $\epsilon$-close to it in diamond norm (see Definition~\ref{def:distanceNorm}). 

\begin{definition}
For any linear operator $E \in \ml{B}(\ml{H})$, its \textbf{Schatten $p$-norm} is defined by $\|E\|_p = \bigl(\Tr[(E^\dagger E)^{p/2}]\bigr)^{1/p}$. 
    For states $\rho, \sigma \in \Dens(\ml{H})$, their trace distance is defined as 
    $$
    d_{\Tr}(\rho, \sigma) = \frac{1}{2} \|\rho - \sigma\|_{1}.
    $$
\end{definition}

\begin{definition}
\label{def:distanceNorm}
For a quantum channel $\ml{E} \in \cptp(\mathcal{H}_A, \mathcal{H}_B)$, its diamond norm is given by
$$
\|\mathcal{E}\|_{\diamond} = \sup_{ \left\|M\right\|_1 \leq 1  } \left\|   
    (\ml{E} \otimes \ml{I}_{\ml{H}_R}) M
    \right\|_{1}.
$$
    The diamond norm distance between two quantum channels $\mathcal{E}, \mathcal{F} \in \cptp(\mathcal{H}_A, \mathcal{H}_B)$ is given by $\|\mathcal{E} - \mathcal{F}\|_{\diamond}$, 
    where the supremum is attainable with a reference space $\ml{H}_R \cong \ml{H}_{A}$. The diamond norm is sub-multiplicative, \ie, $\|\ml{E} \circ \ml{F}\|_{\diamond} \leq \|\ml{E}\|_{\diamond} \|\ml{F}\|_{\diamond}$; it is non-increasing under quantum operations, \ie, for any channel $\ml{E}, \ml{F} \in \cptp(\ml{H}_A, \ml{H}_B)$, $\ml{T} \in \cptp(\mathcal{H}_B, \mathcal{H}_C)$, $\|\ml{T} \circ \ml{E} - \ml{T} \circ \ml{F} \|_{\diamond} \leq \|\ml{E} - \ml{F} \|_{\diamond}$.
\end{definition}

\begin{lemma}[Schur-Weyl duality \cite{Goodman2009, HarrowPhDThesis}]
\label{thm:SchurWeyl}
    Consider the $n$-tensor replication of a Hilbert space $\ml{H}$ with $\dim \ml{H} = d$ and $n \in \mb{N}$. A \textit{partition} $\lambda \vdash n$ of any integer $n \geq 0$ is a tuple $\lambda = (\lambda_1, \dots, \lambda_d)$ such that $\lambda_{1} \geq \lambda_2 \geq \cdots \lambda_{d} \geq 0$ and $\sum_{j=1}^{d} \lambda_j = n$. We use the notion $\lambda \partition{d}{n}$ to indicate $\lambda \vdash n$ into at most $d$ rows. Each $\lambda$ characterizes the shape of a Young diagram. Denote the modules $W_{\lambda}^d$ and $V_{\lambda}$ as the irreducible subspaces of $\U(d)$ and the permutation group $\mathfrak{S}_n$, respectively. The Schur-Weyl duality states that
    $$
    \ml{H}^{\otimes n} \cong \bigoplus_{\lambda \partition{d}{n} } W_{\lambda}^d \otimes V_{\lambda}.
    $$
    Moreover, the $n$-wise tensor product of unitary operator $U \in \U(d)$ admits decomposition
    $$
    U^{\otimes n} \cong \bigoplus_{\lambda \partition{d}{n} } U_{\lambda} \otimes I_{V_{\lambda}},
    $$
    where $(U_{\lambda}, W_{\lambda}^d)$ are irreducible representations of the group $\U(d)$.
\end{lemma}

\begin{definition}[Low-depth brickwork circuit \cite{Zhao_2024, Yang_2025}]
\label{def:brickwork_circuit}
 A brickwork quantum circuit on $N$ qubits consists of sequential operations with $\ell$ unitary gates $\mathscr{G} = \{G_1, \dots, G_{\ell}\}$  arranged across at most $D$ layers. The connectivity among qubits, or the geometry, is given by a connectivity graph $C = ([N], E)$, where $(j, j') \in E$ if local operations are allowed between qubits $j$ and $j'$. Each gate is assumed to be $k$-local with $k = \mathcal{O}(1)$ (bounded constant fan-in), and we denote the set of qubits that $G_j$ acts on as $q_j \subseteq [N]$ subject to $|q_j| = k$ and the subgraph of $C$ induced by qubits in $q_j$, denoted by $C[q_j]$, is connected. Moreover, we denote $\mathcal{L}_r \subseteq [\ell]$ as the indices of gates applied in the $r$-th operation. An illustrative example is given in \autoref{fig:BWCircuit}.

 \begin{figure}[t!]
    \centering
    \includegraphics[width=0.6\linewidth]{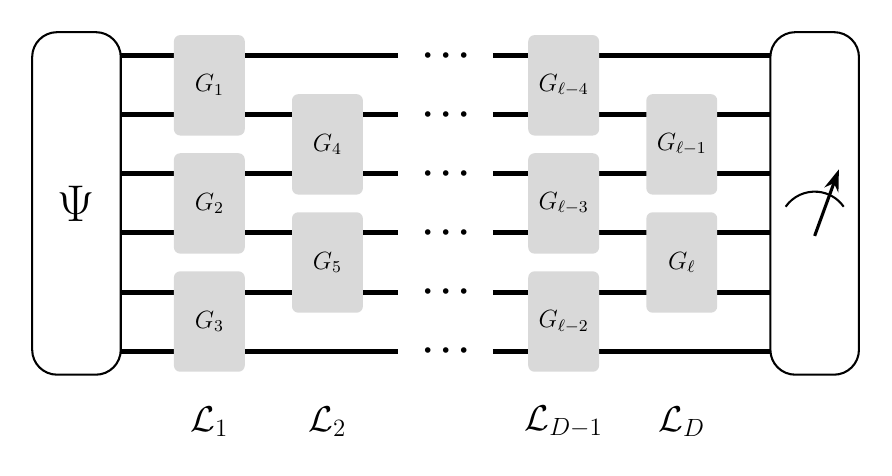}
    \caption{Illustration of a 1D 2-local brickwork quantum circuit consisting of $\ell$ gates arranged across $D$ layers.}
    \label{fig:BWCircuit}
\end{figure}
\end{definition}
    Following the notion of $\mathsf{QNC}$ circuits \cite{Arora2009}, we say a brickwork circuit has low depth if $D = \mathcal{O}(\polylog N)$.

\begin{definition}[Unitary design \cite{MetgerSimpleConstruction2024}]
    Let $\nu$ be any ensemble of unitary operators on $\U(d)$. Its $t$-moment operator is defined as
    $$
    \ml{M}_{\nu}^{(t)}(\rho) := \ee_{U \sim \nu } \left[ U^{\otimes t} \rho {U^{\dagger}}^{\otimes t}  \right].
    $$
    Setting $\nu = \mu$, the Haar measure over $\U(d)$, this is the Haar random $t$-moment operator:
    $
    \ml{M}_{\haar}^{(t)}(\rho) := \int_{\U(d)} U^{\otimes t} \rho {U^{\dagger}}^{\otimes t} \dd \mu(U)
    $. 
    For any $\delta \in (0, 1]$, $\nu$ is a diamond norm $\delta$-approximate unitary $t$-design if $\|\ml{M}_{\nu}^{(t)} - \ml{M}_{\haar}^{(t)}\|_{\diamond} \leq \delta$, and is a relative $\delta$-approximate unitary $t$-design if $(1 -\delta) \mathcal{M}_{\haar}^{(t)} \preceq \mathcal{M}_{\nu}^{(t)} \preceq (1 + \delta) \mathcal{M}_{\haar}^{(t)}$. A $t$-design is also an $s$-design for any $s < t$.
\end{definition}

\begin{remark} 
    On a quantum device, generating Haar random unitaries is inefficient, where the number of gates grows exponentially with the number of qubits \cite{ChristophExact2009}. Therefore, in practice, we often employ the aforementioned unitary $t$-designs to match the first $t$ moments of Haar measures for special computational tasks.
\end{remark}

\section{Approximate programmability and circuit complexity of MO universal unitary programming}
\label{sec:gateComplexityProgramming}
In this section, we comment on the circuit efficiency of the state-of-the-art universal programming scheme, \ie, the learning-based measure-and-operate (MO) programming scheme \cite{BisioOptimalLearning2010, YuxiangOptimal2020}, in terms of its circuit complexity. 

\subsection{Basic setup of universal programming, and the learning-based MO scheme}

\begin{definition}[$\epsilon$-universal quantum processor \cite{YuxiangOptimal2020, Gschwendtner2021programmabilityof}]
\label{def:UniversalProcessor}
    A quantum processor for unitaries in $ \U(d)$ is a tuple $(\mathcal{C}, \{\psi_{U}\}_{U \in \U(d)})$ where $\mathcal{C} \in \cptp(\mathcal{H} \otimes \mathcal{H}_P)$ and $\psi_{U} \in \Dens( \mathcal{H}_P)$ (the `quantum chip'). The space $\mathcal{H}_P$ supports the programmability of the processor, and the processor channel $\mathcal{C}$ is independent of the choice of $U$. The output of the processor when programming $\mathcal{U}$ is given by $\mathcal{E}_U(\cdot) := \Tr_{\mathcal{H}_P}\left[ \mathcal{C}(\cdot \otimes \psi_{U}) \right]$. For any $\epsilon \in (0, 1]$, a processor $\mathcal{E}_U$ is $\epsilon$-universal if
    $$
    \forall \, U \in \U(d), \quad \frac{1}{2} \left\| \mathcal{U} - \mathcal{E}_{U} \right\|_{\diamond} \leq \epsilon.
    $$
\end{definition}

\begin{figure}[t!]
    \centering
    \includegraphics[width=0.7\linewidth]{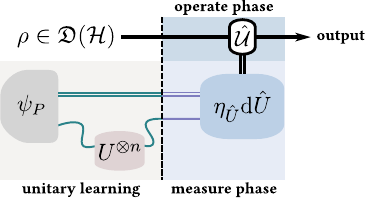}
    \caption{The learning-based MO universal programming scheme.}
    \label{fig:MOScheme}
\end{figure}

The learning-based MO scheme provides a unified framework for programming arbitrary unitary gates in a semi-classical fashion, consisting of \textbf{(1)} coherently learning parallel copies of the unitary gate $U$ with a properly entangled probe state $\ket{\psi_P}$ and storing the information in the quantum register in the form of a program state $\ket{\psi_{P, U}}$; and \textbf{(2)} retrieving the channel from the register via post-selection with a continuous measurement on $\ket{\psi_{P, U}}$, namely the measure-and-operate operation. A diagrammatic illustration is shown in \autoref{fig:MOScheme}. The probe state and the POVM utilize a representation-theoretic design to achieve sample-optimal $\epsilon$-approximate programming at a Heisenberg limit \cite{YuxiangOptimal2020}.
\par  
It is proved in \cite[Lemma 2]{BisioOptimalLearning2010} that due to symmetry, the optimal probe state for learning $n$ parallel copies of an agnostic unitary is of the form $\ket{\psi_P} = \bigoplus_{\lambda \in \mathsf{S}} \sqrt{q_{\lambda}} \ket{\Phi_{W_{\lambda}^d}^+} \otimes \ket{\eta_{\lambda}}$, where $\mathsf{S} \subseteq \mathsf{S}_n^d := \left\{ \lambda~|~\lambda \partition{d}{n} \right\}$ is a subset of all legitimate partitions, $\{q_{\lambda}\}_{\lambda \in \mathsf{S}}$ is a probability distribution, and $\ket{\eta_{\lambda}}$ is an arbitrary bipartite state on $V_{\lambda}^{\otimes 2}$. The learned program state can be expressed by $\ket{\psi_{P, U}} = (U \otimes I)^{\otimes n} \ket{\psi_P}$. To retrieve the unitary, the optimal measurement \cite[Theorem 1]{BisioOptimalLearning2010} is given by $\{ \ket{\eta_{\hat U}}\bra{\eta_{\hat U}} \dd \mu(\hat U) \}$, where $\ket{\eta_{\hat U}} = (\hat U \otimes I)^{\otimes n} \ket{\psi_0}$ and the frame vector $\ket{\psi_0} = \bigoplus_{\lambda \in \mathsf{S}}  \dim W^d_{\lambda} \ket{\Phi_{W_{\lambda}^d}^+} \otimes \ket{\eta_{\lambda}}$. If we use the notation $\chi_{U} := \Tr[U]$ to denote the character of $\U(d)$, the MO processor can be expressed by the following quantum channel:

\begin{equation}
    \label{eqn:MOScheme}
    \begin{aligned}
    \mathcal{E}_{\MO, U}(\rho) &= \int_{\U(d)} \Tr\left( \eta_{\hat U} \psi_{P, U} \right) \cdot \hat{\mathcal{U}}(\rho) \dd \mu(\hat U) \\
    &=  \int_{\U(d)} \left| \sum_{\lambda \in \mathsf{S} } \sqrt{q_{\lambda}} \chi_{\hat{U}_{\lambda} U_{\lambda}^{-1} } \right|^2 \cdot \hat{\mathcal{U}}(\rho) \dd \mu(\hat U)  \\
    &= \mathfrak{p} \cdot \mathcal{U}(\rho) + (1 - \mathfrak{p}) \cdot \frac{I_{\mathcal{H}}}{d}, \\
    \mathfrak{p} &= \frac{1}{d^2 - 1} \left(
    \int_{\U(d)} \left|\sum_{\lambda \in \mathsf{S}} \sqrt{q_{\lambda}} \sum_{\gamma \in \mathsf{O}_1(\lambda) } \chi_{\hat U_{\gamma}} \right|^2 \dd \mu (\hat U) - 1
    \right),
    \end{aligned}
\end{equation}
where the set of partitions $\mathsf{O}_{1}(\lambda) \subseteq \mathsf{S}_{n+1}^d$ is induced by the tensor product of shapes $\lambda \otimes (\Box)$, and $(\Box)$ corresponds to the trivial representation $\hat U_{\Box} = \hat U$. 
Note that the last equality follows from Schur's lemma \cite{Emerson_2005}, since the composite channel $\mathcal{E}_{\MO, U} \circ \mathcal{U}^{\dagger}$ is covariant. The expression of the ``depolarizing'' coefficient $\mathfrak{p}$ is equivalent to the entanglement fidelity \cite[Equation B13]{YuxiangOptimal2020} up to a scalar, which the authors use as an intermediate identity to derive the diamond norm distance between $\mathcal{E}_{\MO, U}$ and $\mathcal{U}$. Alternatively, our formulation in Equation~\ref{eqn:MOScheme} allows direct evaluation.

\subsection{Circuit complexity of unitary programming with the optimal MO scheme}
\label{subsec:CircuitComplexityForUnitaryProgramming}
Although \cite{BisioOptimalLearning2010, YuxiangOptimal2020} have demonstrated the information-theoretic optimality of the learning-based MO programming scheme, its operational feasibility, \ie, circuit complexity, remains unexplored. In the following context, we discuss the circuit complexity of implementing the measure-and-operate phase upon obtaining the program state in the optimal MO programming scheme. Firstly, preparing the observable $\psi_0$ requires applying the Schur transformation to the computational basis of $\mathcal{H}^{\otimes n}$, generating entanglement in the Schur-Weyl basis, and adjusting the amplitudes by a Grover-type algorithm \cite{grover2002creatingsuperpositionscorrespondefficiently}. The Schur transformation can be implemented with $\mathcal{O}\left(\poly(n, \log d, \log(1 / \zeta))\right)$ elementary operations with error $\zeta$ in operator norm \cite{Krovi2019efficienthigh, burchardt2025highdimensionalquantumschurtransforms}, while the subsequent entangling and amplitude-tuning operations require $\Theta(n \log d)$ gates \cite{OptimalStatePreparaition2022}. Suppose that the Schur transformation is $\zeta$-approximate and produces state $\ket{\widetilde{\psi}_0}$. Then $\| \widetilde{\psi}_0 - \psi_0\|_1 \leq \zeta$, and $\left\| \widetilde{\eta}_{\hat U} - \eta_{\hat U} \right\|_1 \leq \zeta$. The corresponding channel $\widetilde{\mathcal{E}}_{\MO, U}$ satisfies
 $$
 \begin{aligned}
\frac{1}{2} \left\| \widetilde{\mathcal{E}}_{\MO, U} - \mathcal{E}_{\MO, U} \right\|_{\diamond} &\leq \frac{1}{2}\int_{\U(d)} \left| \Tr((\widetilde{\eta}_{\hat U} - \eta_{\hat U}) \psi_{P, U}) \right| \dd \mu(\hat U) \\
 &\leq \int_{\U(d)}  d_{\Tr}\left(\widetilde{\eta}_{\hat U}, \eta_{\hat U} \right)  \dd \mu(\hat U) \\
 &\leq \frac{1}{2}\zeta.
 \end{aligned}
 $$
 One can readily verify that this inequality holds for any ensemble other than the Haar measure. For simplicity, we restrict $\zeta = \mathcal{O}(\epsilon)$ and ignore this error term in the latter context. 
 \par A major computational burden lies in applying the optimal POVM to the program state. Mathematically, this is equivalent to first using Haar random resources to generate the matrix $\hat U^{\otimes n}$, use it to synthesize the measurement operator $\eta_{\hat U}$, and perform the two-outcome PVM $\{\eta_{\hat U}, I - \eta_{\hat U}  \}$. The quantum processor applies $\hat{\mathcal{U}}$ to the input state if the outcome ``$\eta_{\hat U}$'' occurs, and acts trivially (aborts) otherwise. Therefore, the circuit complexity of implementing the POVM is closely tied to the generation\footnote{For the generation of the $n$-tensor $\hat U^{\otimes n}$, the circuit depth is unaffected, as we can generate identical gates with a fixed configuration on $n$ sites in parallel.} of $\hat U$. In practice, we might use unitary designs to mitigate the circuit depth requirement. We evaluate the robustness of the MO scheme's performance in terms of unitary design accuracy in the following context, before which we show a matrix inequality lemma.

 \begin{lemma}
 \label{lemma:matrixInequality}
     For two CPTP maps $\mathcal{A}, \mathcal{B} \in \cptp(\mathcal{H})$ that satisfy $(1 - \upsilon) \mathcal{A} \preceq \mathcal{B} \preceq (1 + \upsilon) \mathcal{A}$ for some $\upsilon \in [0, 1)$, the tensor product Hilbert space $\mathcal{H} = \mathcal{H}_S \otimes \mathcal{H}_0$,  any Hermitian and positive semidefinite operators $X$, $Y = Y_{S} \otimes I_0 \in \mathcal{B}(\mathcal{H})$, it holds that
     $$
     \left\| \Tr_{\mathcal{H}_S}\left[ \left(\mathcal{A} - \mathcal{B}\right)(X)  Y  \right]  \right\|_{1} \leq \upsilon  \left\| \Tr_{\mathcal{H}_S} \left[ \mathcal{A}(X) Y \right]  \right\|_1.
     $$
 \end{lemma}

 \begin{proof}
    Note that $(1 - \upsilon) \mathcal{A} \preceq \mathcal{B} \preceq (1 + \upsilon) \mathcal{A}$ implies
    \begin{equation}
    \label{eqn:LowenerOrder}
   -\upsilon \mathcal{A}(X) \preceq \left(\mathcal{A} - \mathcal{B}\right)(X) \preceq \upsilon \ml{A}(X).
    \end{equation}
    Denote $Z = (\mathcal{A} - \mathcal{B})(X)$, $W = \mathcal{A}(X)$. Using \cite[Equation 9.22]{Nielsen2012} that $\|F\|_1 = \sup_{Q: \|Q\| \leq 1} \Tr\left[ FQ \right]$, we can rewrite the expression of the trace norm on both sides:
    $$
    \begin{aligned}
        \left\| \Tr_{\mathcal{H}_S}\left[ Z  Y  \right]  \right\|_{1}  &= \sup_{Q:\|Q\| \leq 1 } \Tr\left[  \Tr_{\ml{H}_S}\left[Z  Y \right]  Q  \right] \\
        &= \sup_{Q: \|Q\| \leq 1} \Tr\left[ (Z  Y) (I_S \otimes Q) \right] \\
        &= \sup_{Q: \|Q\| \leq 1} \Tr\left[ Z  \left( Y_S \otimes Q  \right) \right].
    \end{aligned}
    $$
    Since $Y_S \succeq 0$, we take its square root $\sqrt{Y_S}$. Denote $\widetilde{Z} = \left(\sqrt{Y_S} \otimes I_0\right) Z \left(\sqrt{Y_S} \otimes I_0 \right)$, it holds that $\Tr\left[ Z  \left( Y_S \otimes Q  \right) \right] = \Tr\left[ \widetilde{Z}  (I_S \otimes Q)  \right]$. Analogously, we have $\Tr\left[ W  (Y_S \otimes Q) \right] = \Tr\left[ \widetilde{W}  (I_S \otimes Q)  \right]$ if we denote $\widetilde{W} = \left(\sqrt{Y_S} \otimes I_0\right) W \left(\sqrt{Y_S} \otimes I_0 \right)$. The condition stated in Equation~\ref{eqn:LowenerOrder} implies $-\upsilon \widetilde{W} \preceq \widetilde{Z} \preceq \upsilon \widetilde{W}$, and therefore $\left|\Tr_{\mathcal{H}_S}(\widetilde{Z}) \right| \preceq \Tr_{\mathcal{H}_S}(|\widetilde{Z}|) \preceq \upsilon \Tr_{\mathcal{H}_S}(\widetilde{W})$\footnote{
    By the fact that the partial trace operation is a completely positive map \cite{Watrous_2018}.
    }. If we substitute in $\widetilde{Z}$ in the supremum, we have
    $$
    \sup_{Q: \|Q\| \leq 1} \Tr\left[ \widetilde{Z} (I_S \otimes Q)  \right] = \sup_{Q: \|Q\| \leq 1} \Tr\left[ \Tr_{\mathcal{H}_S}(\widetilde{Z})  Q \right] = \left\| \Tr_{\mathcal{H}_S}(\widetilde{Z}) \right\|_1,
    $$
    while the same holds for $\widetilde{W}$.
    Combining the previous statements and that $\|M\|_1 = \||M|\|_1$ for any operator $M$, we have
    $$
    \begin{aligned}
    \left\| \Tr_{\mathcal{H}_S}\left[ \left(\mathcal{A} - \mathcal{B}\right)(X)  Y  \right]  \right\|_{1} &= \sup_{Q: \|Q\| \leq 1} \Tr\left[\widetilde{Z}  (I_S \otimes Q) \right] = \left\| \Tr_{\mathcal{H}_S}(\widetilde{Z}) \right\|_1 = \left\| \left| \Tr_{\mathcal{H}_S}(\widetilde{Z}) \right| \right\|_1 \\
    &\leq \upsilon \left\| \Tr_{\mathcal{H}_S} (\widetilde{W}) \right\|_1 = \upsilon \sup_{Q: \|Q\| \leq 1} \Tr\left[ \widetilde{W} (I_S \otimes Q)  \right] \\
    &= \upsilon  \left\| \Tr_{\mathcal{H}_S} \left[ \mathcal{A}(X)  Y \right]  \right\|_1.
    \end{aligned}
    $$
    This completes the proof.
 \end{proof}

\begin{theorem}
\label{thm:ApproximateMOScheme}
    A relative $\delta$-approximate unitary $(n+1)$-design implements the measure phase $\delta$-approximately using the MO scheme that learns $n$ copies of $U$, assuming the observable $\psi_0$ is prepared perfectly. 
\end{theorem}

\begin{proof}
    Note that the channel $\mathcal{E}_{\MO, U}$ can be rewritten as
    $$
    \begin{aligned}
        \mathcal{E}_{\MO, U} &= \int_{\U(d)} \Tr\left( \eta_{\hat U} \psi_{P, U} \right) \cdot \hat{\mathcal{U}} \dd \mu(\hat U) \\
        &= \int_{\U(d)} \Tr\left( (\hat U \otimes I)^{\otimes n} \psi_0 { (\hat U \otimes I)^{\otimes n} }^{\dagger} \psi_{P, U} \right) \cdot \hat{\mathcal{U}} \dd \mu(\hat U).
    \end{aligned}
    $$
    When replacing the Haar random ensemble with a $\delta$-approximate ensemble, we first sample $S \sim \nu_{\delta}$ and apply the $n$-wise tensor to obtain the POVM operator. Denote the $n$ ancilla registers into which the program state is inserted by $\mathsf{A}_1, \dots, \mathsf{A}_n$, for any state $\rho \in \Dens(\mathcal{H})$,
    \begin{equation}
    \label{eqn:IntegrandReformulation}
    \begin{aligned} &\Tr\left( (\hat U^{\otimes n} \otimes I) \psi_0 ( {\hat{U}^{\otimes n}} \otimes I )^{\dagger} \psi_{P, U} \right) \cdot \hat{\mathcal{U}}(\rho)  
    \\ = \,& \Tr_{\mathsf{A}_1, \dots, \mathsf{A}_n} \left[ (\hat U \otimes I)^{\otimes n} \psi_0 { (\hat U \otimes I)^{\otimes n} }^{\dagger} \psi_{P, U} \otimes \hat{\mathcal{U}}(\rho) \right] \\
    = \, & \Tr_{\mathsf{A}_1, \dots, \mathsf{A}_n} \left[ \left(\left(\left(\hat{\mathcal{U}} \otimes \mathcal{I}\right)^{\otimes n} \otimes \hat{\mathcal{U}}\right)(\psi_0 \otimes \rho) \right) \left(\psi_{P, U} \otimes I \right)  \right].
    \end{aligned}
    \end{equation}    
    For clarity, we append the channel with a superscript to indicate the underlying unitary ensemble, \ie, $\mathcal{E}_{\MO, U}^{\nu}$ when $\nu$ is utilized. It will be convenient to define the channel
    $$
    \mathcal{Q}_{\nu}^{(t)}(\rho) = \ee_{U \sim \nu } \left[
    \left( \mathcal{U} \otimes \mathcal{I}  \right)^{\otimes t}(\rho)
    \right] 
    $$
    for any unitary ensemble $\nu$ and $t \in \mathbb{N}$. Observe that $\mathcal{Q}_{\nu}^{(t)} \cong \mathcal{M}_{\nu}^{(t)} \otimes \mathcal{I}^{\otimes t}$, it follows that $(1 - \delta) \mathcal{Q}_{\haar}^{(t)} \preceq \mathcal{Q}_{\nu}^{(t)} \preceq (1 + \delta) \mathcal{Q}_{\haar}^{(t)}$ when $\nu$ is a relative $\delta$-approximate unitary $t$-design. Denote the auxiliary ancillae in the definition of diamond norm by $\mathsf{R}$, and the working register by $\mathsf{O}$, we obtain
    $$
    \begin{aligned}
        \left\| \mathcal{E}_{\MO, U}^{\nu_{\delta}} - \mathcal{E}_{\MO, U}^{\haar}  \right\|_{\diamond} &= \max_{\ket{\rho_{\mathsf{OR}}}} \left\| \left((\mathcal{E}_{\MO, U}^{\nu_{\delta}} - \mathcal{E}_{\MO, U}^{\haar}) \otimes \mathcal{I}_{\mathsf{R}} \right)(\rho_{\mathsf{OR}})  \right\|_{1}.
    \end{aligned}
    $$
    Note that with the equality in Equation~\ref{eqn:IntegrandReformulation} and that $\mathcal{H}_{\mathsf{R}} \cong \mathcal{H}$, 
    $$
    \begin{aligned}
    &\left\| \left((\mathcal{E}_{\MO, U}^{\nu_{\delta}} - \mathcal{E}_{\MO, U}^{\haar}) \otimes \mathcal{I}_{\mathsf{R}} \right)(\rho_{\mathsf{OR}})  \right\|_{1} \\
    =~&\left\|  \ee_{S \sim \nu_{\delta}}\left[\Tr\left( \eta_{S} \psi_{P, U} \right) \cdot (\mathcal{S} \otimes \mathcal{I}_{\mathsf{R}})(\rho_{\mathsf{OR}}) \right] - \int_{\U(d)} \Tr\left( \eta_{\hat U} \psi_{P, U} \right) \cdot (\hat{\mathcal{U}} \otimes \mathcal{I}_{\mathsf{R}})(\rho_{\mathsf{OR}}) \dd \mu(\hat U)    \right\|_1 \\
    =~& \left\|  \Tr_{\mathsf{A}_1, \dots, \mathsf{A}_n} \left[ \left(\left(\mathcal{Q}_{\nu_{\delta}}^{(n+1)} - \mathcal{Q}_{\haar}^{(n+1)}\right)(\psi_0 \otimes \rho_{\mathsf{OR}})  \right) \left(\psi_{P, U} \otimes I_{\mathsf{OR}} \right)  \right] \right\|_{1} \\
    \leq ~&~\delta \cdot \left\|  \Tr_{\mathsf{A}_1, \dots, \mathsf{A}_n} \left[  \left(  \mathcal{Q}_{\haar}^{(n+1)}(\psi_0 \otimes \rho_{\mathsf{OR}}) \right) \left(\psi_{P, U} \otimes I_{\mathsf{OR}} \right) \right] \right\|_{1} \\
    =~&~\delta \cdot \left\| \left(\mathcal{E}_{\MO, U}^{\haar} \otimes \mathcal{I}_{\mathsf{R}}\right)(\rho_{\mathsf{OR}})  \right\|_1 \\
    =~&~ \delta,
    \end{aligned}
    $$
    where we have used Lemma~\ref{lemma:matrixInequality} in the last inequality. This completes the proof.
    \end{proof}

\begin{table*}[t!]
    \centering
    \begin{tabular}{>{\centering\arraybackslash}p{0.8cm}  >{\centering\arraybackslash}p{5.3cm} >{\centering\arraybackslash}p{3.9cm} > {\centering\arraybackslash}p{3.3cm}}
    \toprule
        \textbf{Ref.} & \textbf{Depth} &\textbf{Condition} &\textbf{Type}\\
        \midrule
        \cite{Harrow2023} & $\mathcal{O}\left(\poly\left(t, \log (1 / \varrho) \right) \cdot N^{1/\mathfrak{K}} \right)$ & $C$ is a $\mathfrak{K}$-lattice &diamond \\
        \cite{Jeongwan2024} & $\mathcal{O}\left(\left( Nt^2 + t \log (1/\varrho)  \right) \log N  \right)$ & $-$ &diamond \& relative \\ \cite{MetgerSimpleConstruction2024} &$\mathcal{O}\left( t \, \poly N + t \log (1 / \varrho) \right)$ &$t \leq 2^{N/4}$ &diamond\\
        \cite{MetgerSimpleConstruction2024} &$\mathcal{O}\left( t^2 \, \poly N + t^2 \log (1 / \varrho) \right)$ &$t \leq 2^{N/4}$ &relative
        \\
        \cite{chen2024incompressibilityspectralgapsrandom} &$\mathcal{O}\left( \left( Nt + 
        \log(1 / \varrho)  \right) \log^7 t \right)$ &$t = \mathcal{O}\left(2^{2N/5} \right)$ &diamond \& relative\\
        \cite{schuster2025randomunitariesextremelylow} &$ \mathcal{O} \left( \left(  \xi t + \log\left(N / \varrho \right) \right) \log^7 t  \right) $ &$t = \mathcal{O}\left(2^{2\xi / 5}\right)$, $\exists \, \xi \geq 1$ &diamond \& relative\\
    \bottomrule
    \end{tabular}
    \caption{Circuit depth upper bound for $\varrho$-approximate unitary $t$-designs on $N$ qubits.}
    \label{tab:kDesignCircuitDepth}
\end{table*}

For an $\epsilon$-universal quantum processor constructed from the MO scheme, by a simple additive argument, $ \frac{1}{2} \| \mathcal{E}_{\MO, U}^{\nu_{ \delta}} - \mathcal{U} \|_{\diamond} \leq \epsilon +  \frac{1}{2}  \delta$. 
To ensure the performance of the retrieval of $U$ from the program state, we require $\delta = \mathcal{O}(\epsilon )$. Note that a sufficiently accurate $(n+1)$-design is necessary, as a large deviation from the optimal covariant POVM \cite{ChiribellaOptimalEst2005, BisioOptimalLearning2010, YuxiangOptimal2020} affects the density of the measurement outcomes, thus reducing the retrieval precision. Moreover, the query complexity of either learning or programming a $d$-dimensional unitary scales as $n = \Omega(d^2)$ \cite{YuxiangOptimal2020, HaahQueryOptimal2023, Zhao_2024}, forcing the design to match the Haar measure to a high order\footnote{
Most constructions of unitary $t$-designs are gate-efficient only when $t = \mathcal{O}\left(2^{N/2}\right)$.
}. As presented in \autoref{tab:kDesignCircuitDepth}, the circuit depth bound becomes trivial as $t \sim 2^{2N}$. Therefore, the circuit depth requirement for implementing the unitary design in the measure phase is almost identical to that of implementing a genuine Haar measure, which requires an elementary gate sequence with depth  $d^2$ \cite{QuantumLogicCircuitSynthesis2005}, resulting in an $\mathcal{O}(nd^2 \log d + \poly(n, \log d, \log(1 / \zeta)))$ gate complexity in implementing the Schur-transformed $n$-wise unitary $\hat{U}^{\otimes n}$. 
Upon obtaining the measurement outcome, it suffices to construct the unitary transformation $\hat{\mathcal{U}}$ from its classical description. We need extra $\mathcal{O}(d^2 \log^3(  d^2 / \tau))$ elementary operations to synthesize it up to $\tau = \mathcal{O}(\epsilon)$ in diamond norm error\footnote{
The diamond norm and the operator norm are equivalent for unitary channels, since $\|U - V\| \leq \|\mathcal{U} - \mathcal{V}\|_{\diamond} \leq 2\|U - V\|$ \cite{Zhao_2024}.
} by the Solovay-Kitaev theorem \cite{kitaev2002classical}. The POVM using a $(n+1)$-design ensemble $\mathsf{X}_{n+1, d}$ requires $\log_2 |\mathsf{X}_{n+1, d}| = \mathcal{O}(d^2 \log n)$ ancillae \cite{Roy_2009}. Therefore, the overall gate complexity scales as $\poly(n, d, \log(1/\zeta), \log(1/\tau), \log(1/\delta))$, with $\mathcal{O}(d^2 \log n)$ ancillae.

\par Finally, physically implementing the MO scheme yields an $\epsilon_{\mathsf{MO}}$-universal processor, where $\epsilon_{\MO}$ is composed of the following terms:
$$
\epsilon_{\mathsf{MO}} = \underbrace{\epsilon}_{\text{optimal retrieval error}} + \underbrace{\frac{1}{2}\zeta}_{ \text{POVM application error} } + \underbrace{\frac{1}{2}  \delta }_{ \text{measure phase error} } + \underbrace{\tau}_{\text{operate phase error}},
$$
and uses $\widetilde{\mathcal{O}}\left(\poly(d, 1 / \sqrt{\epsilon}) \right)$ quantum gates and $\mathcal{O}\left(d^2 \log(d^2 / \epsilon)\right)$ ancillae beyond the queries to the unknown unitary if we take $\zeta, \tau, \delta \sim \epsilon$ and set $n \sim d^2 / \sqrt{\epsilon}$ \cite[Theorem 2]{YuxiangOptimal2020}. A lower circuit complexity is reported in \cite[Theorem 1.1.3]{HaahQueryOptimal2023}; however, the $\poly(d)$ factor remains unavoidable, which is prohibitive for the purpose of designing an operationally efficient universal programming scheme for $\U(d)$. To the best of our knowledge, no gate-efficient algorithm has been proposed. Although there are unitaries that require deeper circuits to generate \cite{Nielsen2012}, recovering a unitary $U$ from the compact program state $\psi_{P, U}$ is generally more challenging than generating it \cite{Haferkamp2022}.

\subsection{Overview of circuit complexity of unitary programming with PbT scheme}
\label{subsec:pbt_based_circuit_complexity}
To compare the circuit complexity analysis for the MO scheme in Section~\ref{subsec:CircuitComplexityForUnitaryProgramming}, we briefly review the circuit complexity of implementing unitary programming with the PbT scheme. As is remarked in \cite{yoshida2025quantumadvantagestorageretrieval}, the PbT scheme is a coherent programming strategy, while the MO scheme breaks the coherence via the covariant measurement. Compared to the MO scheme, the PbT scheme requires an even larger amount of multi-copy entanglement between spatial registers \cite{Mozrzymas2021optimalmultiport}. The computational burden of the PbT scheme also stems from measurement: A Pretty-Good Measurement (PGM) is applied according to the average density matrix over all ports. The primitive Ishizaka-Hiroshima protocol \cite{AsymptoticTeleporation2008} na\"ively applies the raw PGM and yields an $\mathcal{O}(d^{2n})$ circuit complexity, where $n$ is the number of queries to $U$. Subsequent works \cite{fei2023efficientquantumalgorithmportbased, grinko2024efficientquantumcircuitsportbased, EfficientAlgoForAllPbT2024} have introduced various methods to reduce circuit complexity by exploiting the symmetry across different ports. The asymptotic scalings of the complexity of these optimized approaches are consistent with those established in Section~\ref{subsec:CircuitComplexityForUnitaryProgramming}, all requiring a gate complexity of $\mathcal{O}(\poly(n, d))$ and an ancillae complexity of $\mathcal{O}(d^2 \log n)$.

\subsection{Efficient programming of low-depth brickwork circuits with MO scheme}

Although programming an arbitrary unitary by preparing and retrieving a large programming state is resource-intensive, programming a low-depth brickwork circuit can be efficient. Recall from Definition~\ref{def:brickwork_circuit} that the circuit can be decomposed into local unitary operations that apply either sequentially or in parallel. If we approximately program each unitary up to a small error, the whole circuit can be programmed approximately with a satisfying performance, given that the number of gates is bounded. The following statement characterizes the error propagation across a brickwork circuit.

\begin{fact}[\cite{Nielsen2012}]
\label{fact:errorPropagation}
    For a brickwork circuit on $\ell$ gates, if each gate is programmed up to error $\epsilon$ in diamond norm, then the unitary generated by the circuit is programmed up to error $\ell \epsilon$ in diamond norm.
\end{fact}

Since the local gates on the brickwork circuit have a constant dimension $2^k = \mathcal{O}(1)$, using the results from Section~\ref{subsec:CircuitComplexityForUnitaryProgramming}, it requires only $\widetilde{\mathcal{O}}(\poly (1/\sqrt{\epsilon'}))$ quantum gates to program each gate $\epsilon'$-approximately from the collection of their program state by invoking the MO scheme as a subroutine. As per Fact~\ref{fact:errorPropagation}, to program a low-depth circuit $\epsilon$-approximate in diamond norm, setting $\epsilon' = \epsilon / \ell$ would suffice. Therefore, programming a brickwork circuit with a fixed architecture $\epsilon$-close in diamond norm requires 
$$
\ell \cdot \widetilde{\mathcal{O}}\left( \poly( 1 / \sqrt{\epsilon'} ) \right) = \widetilde{\mathcal{O}}\left( \ell \cdot \poly( \sqrt{\ell / \epsilon} ) \right) = \widetilde{\mathcal{O}}\left( \poly \left(N, 1 / \sqrt{\epsilon} \right)  \right),
$$
quantum gates and $\mathcal{O}(\ell \log (1 / \sqrt{\epsilon'}) ) = \widetilde{\mathcal{O}}(N \log (1 / \epsilon) )$ space, by noting that $\ell \leq \frac{ND}{k} = \mathcal{O}(N \polylog N )$.

\begin{remark}
    Low-depth unitary programming via PbT schemes yields an analogous circuit complexity, as detailed in Section~\ref{subsec:pbt_based_circuit_complexity}. 
\end{remark}

\section{Bounds for quantum memory complexity of programming low-depth circuits}

Having discussed the hardness of recovering information about general unitaries from the quantum register in the programming scheme, we are also concerned with how large the register should be to achieve approximate programmability. Precisely, we utilize the program states $\ket{\psi_{P, U}}$ [cf. Section~\ref{sec:gateComplexityProgramming}] to encode the parallel queries to $\mathcal{U}$, henceforth store them in the quantum memory before we ever call the processor to program the unitary. Mathematically, the \textit{program dimension} $d_P$ is defined as the dimension of the subspace where the program states are supported:
$$
d_P := \dim \left( \overline{\Span}\left\{ \psi_{P, U}~|~U \in \U(d) \right\}  \right),
$$
where $\overline{\Span}$ is the closure of the spanned subspace. 
This identity quantifies the (quantum) memory capacity required to store these quantum states. If we restrict to approximate programmability, there is a finite set of program states, and thus we can remove the closure operation. Equivalently, the base-2 logarithm of the dimension $c_P = \log_2 d_P$, or the \textit{program cost}, quantifies how many qubits are needed in the memory, which is the main figure of merit of \cite{Kubicki2019, YuxiangOptimal2020, Gschwendtner2021programmabilityof, yoshida2025quantumadvantagestorageretrieval}. Although prior works have reported either upper and lower bounds for the program cost \cite{ProbabilisticUniversalProcessor2002, MBQComp2007, AsymptoticTeleporation2008, Kubicki2019, YuxiangOptimal2020, Gschwendtner2021programmabilityof, Muguruza2024portbasedstate, yoshida2025quantumadvantagestorageretrieval}, they are non-trivial only when the quantum circuit resides on a constant number of qubits, \ie, $N = \mathcal{O}(1)$. However, low-depth quantum circuits often extend their registers to preserve their computational power \cite{Bravyi2018}. Therefore, we will instead derive the bounds in the large $N$ regime, while the feasible programming error $\epsilon$ yields a lower bound related to the circuit architecture.

\subsection{Lower bound for the program cost}
\label{subsec:CostLowerbound}
We start by presenting several useful lemmas.

\begin{lemma}[{\cite[Appendix A]{YuxiangOptimal2020}}]
\label{lemma:informationCapacity}
    For any $\epsilon$-universal processor $(\ml{C}, \{\psi_{P, U}\}_{U \in \U(d)})$, where $\ml{C}$ is some operations that constructs $\ml{U}$ from each $\psi_{P, U}$, for each $U \in \U(d)$, there exists a channel $\ml{P}_U(\cdot) = \ml{K}_{\ml{C}}  ((\cdot) \otimes \psi_{P, U}) \in \cptp(\ml{H}^{\otimes 2n})$ that acts trivially on the multiplicity subspace, where $\ml{K}_{\ml{C}}$ is a $\ml{C}$-dependent quantum operation, such that
    $$
    \left\| \ml{P}_U - (\ml{U} \otimes \mathcal{I})^{\otimes n} \right\|_{\diamond} \leq 4n\sqrt{2\epsilon}.
    $$
    
\end{lemma}

We will derive the lower bound with an information-theoretic approach, by quantifying the Holevo information of an ensemble generated by a low-depth quantum circuit.

\begin{lemma}[\cite{holevo1973bounds, Nielsen2012}]
\label{lemma:HolevoInformation}
    For any quantum state $\rho \in \Dens(\mathcal{H})$, its \textit{von Neumann entropy} is given by $S(\rho) = -\Tr\left(\rho \log \rho\right)$. For an ensemble of quantum state $\left\{ \rho_x \dd x \right\}_{x \in \mathcal{X}}$, its \textit{Holevo information} is given by
    $$
    \chi\left( \left\{ \rho_x \dd x \right\}_{x \in \mathcal{X}} \right) = S\left( \int_{\mathcal{X}} \rho_x \dd x  \right) - \int_{\mathcal{X}} S(\rho_x) \dd x.
    $$
    Suppose the ensemble is supported on a $d_{\mathcal{X}}$-dimensional subspace; then its Holevo information is bounded above by
    $$
    \chi\left( \left\{ \rho_x \dd x \right\}_{x \in \mathcal{X}} \right) \leq \log d_{\mathcal{X}}.
    $$
    The data processing inequality holds for the Holevo information, \ie, $$
    \chi\left( \left\{ \mathcal{E}(\rho_x) \dd x \right\}_{x \in \mathcal{X}}  \right) \leq \chi\left(\left\{ \rho_x \dd x \right\}_{x \in \mathcal{X}}\right), \quad \forall \, \mathcal{E} \in \cptp(\mathcal{H}).
    $$
\end{lemma}

\begin{lemma}[Alicki–Fannes–Winter \cite{Winter2016}, reformulated]
\label{lemma:AFW}
    For any quantum states $\rho, \sigma \in \Dens(\mathcal{H})$ that differ on a $\mathfrak{d}$-dimensional subspace of $\mathcal{H}$, it holds that
    $$
    \left|S(\rho) - S(\sigma)\right| \leq  \log \mathfrak{d} \cdot d_{\Tr}(\rho, \sigma) + \log 2.
    $$
\end{lemma}

\begin{lemma}
\label{lemma:BinomialLowerbound}
    For any $m, k \in \mathbb{N}$, it holds that 
    $$
    \binom{m + k}{k} \geq \frac{1}{m + k + 1} \left( 1 + \frac{k}{m+1} \right)^{m+1} \left( 1 + \frac{m}{k+1} \right)^{k+1}.
    $$
\end{lemma}

\begin{proof}
    Taking the logarithm of the left-hand side, we have 
    $$
    \log \binom{m + k}{k} = \log \prod_{j=1}^{k} \left(1 + \frac{m}{j} \right) = \sum_{j=1}^{k} \log \left( 1 + \frac{m}{j} \right).
    $$
    Since the function $\log\left(1 + \frac{m}{x} \right)$ is convex on $\mathbb{R}_{> 0}$, it holds that 
    $$
    \begin{aligned}
        \sum_{j=1}^k \log \left(1 + \frac{m}{j} \right) & \geq \int^{k}_0 \log \left(1 + \frac{m}{x + 1} \right) \dd x \\
        &= m \log \left( \frac{m + x}{e} \right) + x \log \left(1 + \frac{m}{x} \right) \bigg|_{1}^{k+1} \\
        &= (k+1) \log\left(1 + \frac{m}{k+1} \right) + (m+1) \log\left(1 + \frac{k}{m + 1} \right) - \log(m + k + 1).
    \end{aligned}
    $$
    Taking the exponential on both sides of the inequality yields the desired result.
\end{proof}

Now we are ready to present the lower bound for the program cost with bounded programming error.

\begin{theorem}
\label{thm:DimensionLowerbound}
    Programming a quantum circuit on $N$ qubits and depth $D = \mathcal{O}(\polylog N)$ up to error $\epsilon = \Omega\left(1 / \polylog N \right) < \frac{1}{32}$ in diamond norm distance requires a quantum processor with program cost $c_P$ that satisfies
    $$
    c_P \geq { \varpi \left( 1 - \frac{\kappa}{2} \right)^2\left(\frac{1 - \varpi}{4\sqrt{2\epsilon}} - 1\right) } \log_2 \frac{ 4e\sqrt{2\epsilon} d }{(1 - \kappa/2) \varpi } - c_0
    $$
    for any $\varpi \in (0, 1 - 4\sqrt{2\epsilon})$, $\kappa = \Omega(2^{-\polylog(N)}) < 1$, and the constant $c_0 = 5 + \frac{1}{2 \log 2} \approx 5.72$. 
\end{theorem}

\begin{proof}
    Suppose we have a diamond norm $\kappa$-approximate unitary $n$-design $\nu_{\kappa}$, where the error parameter $\kappa = \Omega(2^{-\polylog N})$ and $n = \lceil \frac{(1- \kappa/2)\varpi}{4\sqrt{2\epsilon}} \rceil = \mathcal{O}(1 / \sqrt{\epsilon}) = \mathcal{O}(\polylog N)$ for some constant $\varpi$ such that $0 < \varpi < 1 - 4\sqrt{2\epsilon}$. Using Schuster, Haferkamp, and Huang's construction \cite{schuster2025randomunitariesextremelylow} [cf. \autoref{tab:kDesignCircuitDepth}], the unitary design can be generated on a 1D circuit with $\mathcal{O}(\polylog N)$ depth, which is generalizable to any circuit geometry with the same order of depth \cite[Appendix D.1]{schuster2025randomunitariesextremelylow} using local SWAP operations. Therefore, the Holevo information of an ensemble of states generated by $n$ parallel uses of the unitary design on a fixed initial state presents a lower bound for the information required to program our low-depth circuits. Our derivation is based on analyzing the amount of information in the collection of program states
    $
    \chi\left( \left\{ \psi_{P, V} \dd \nu_{\kappa}(V) \right\} \right)
    $.
    Using the property of the Holevo information, when inserting any quantum state $\Phi_n$ into the channel introduced in Lemma~\ref{lemma:informationCapacity}, it holds that
    \begin{equation}
    \label{eqn:HolevoInequality_1}
    \begin{aligned}
        \chi\left( \left\{ \mathcal{P}_V(\Phi_n) \dd \nu_{\kappa}(V) \right\} \right) &= \chi\left( \left\{ \mathcal{K}_{\mathcal{C}}\left( \Phi_n \otimes \psi_{P, V} \right) \dd \nu_{\kappa}(V) \right\} \right) \\
        &\leq \chi\left( \left\{  \Phi_n \otimes \psi_{P, V}  \dd \nu_{\kappa}(V) \right\} \right) \\
        &=  \chi\left( \left\{   \psi_{P, V}  \dd \nu_{\kappa}(V) \right\} \right).
    \end{aligned}
    \end{equation}
    We first note that for any $U \in \U(d)$, both $\mathcal{P}_{U}$ and $\mathcal{U}^{\otimes n}$ act trivially on the symmetric subspaces. Therefore, on the same input state, their outputs only differ on the subspaces that correspond to the irreducible representations of $\U(d)$, which admit dimension
    $$
    d_{n} = \sum_{\lambda \partition{d}{n} } (\dim W_{\lambda}^d)^2.
    $$
    Using the inequalities in Lemma~\ref{lemma:informationCapacity}, \ref{lemma:HolevoInformation} and \ref{lemma:AFW}, for an $\epsilon$-universal quantum processor,
    \begin{equation}
    \label{eqn:HolevoInequality_2}
    \begin{aligned}
   &~\left| \chi\left( \left\{ \ml{P}_V(\Phi_n) \dd \nu_{\kappa}(V) \right\} \right) - \chi\left( \left\{   (\mathcal{V} \otimes \mathcal{I})^{\otimes n}(\Phi_n) \dd \nu_{\kappa}(V) \right\} \right) \right| \\
   \leq &~\left| S \left( \ee_{V \sim \nu_{\kappa} } \left[ \ml{P}_V(\Phi_n) \right]  \right) - S\left( \ee_{V \sim \nu_{\kappa} } \left[ (\mathcal{V} \otimes \mathcal{I})^{\otimes n}(\Phi_n) \right]  \right) \right| \\
   & \quad + \left| \ee_{V \sim \nu_{\kappa} } \left[ S\left( \ml{P}_V(\Phi_n)\right)  - S\left( (\mathcal{V} \otimes \mathcal{I})^{\otimes n}(\Phi_n) \right)\right] \right| \\
   \leq &~\frac{\log d_n}{2} \left\| \ee_{V \sim \nu_{\kappa} } \left[ \left( \ml{P}_V - (\mathcal{V} \otimes \mathcal{I})^{\otimes n} \right)(\Phi_n) \right]  \right\|_{1} + \log 2 \\
   & \quad + \ee_{V \sim \nu_{\kappa} } \left[ \frac{\log d_n}{2} \left\|  \left( \ml{P}_V - (\mathcal{V} \otimes \mathcal{I})^{\otimes n} \right)(\Phi_n)   \right\|_{1} + \log 2 \right] \\
   \leq&~  \log d_n \cdot \ee_{V \sim \nu_{\kappa} } \left[ \left\|\left( \ml{P}_V - (\mathcal{V} \otimes \mathcal{I})^{\otimes n} \right)(\Phi_n)\right\|_1 \right] + 2 \log 2 \\
   \leq&~ 4 n \sqrt{2\epsilon} \log d_n + 2 \log 2.
   \end{aligned}
    \end{equation}
    Furthermore, we consider the ensemble $\chi\left( \left\{   (\mathcal{U} \otimes \mathcal{I})^{\otimes n}(\Phi_n) \dd \mu(U) \right\} \right)$, where $\mu$ is the Haar measure on $\U(d)$. Analogous to the formulation in Equation~\ref{eqn:HolevoInequality_2}, if we take $\Phi_n = \ket{\Phi_n} \bra{\Phi_n}$ as a pure state, the von Neumann entropy in the second term vanishes, and thus
    \begin{equation}
    \label{eqn:HolevoInequality_3}
    \begin{aligned}
    &\left| \chi\left( \left\{   (\mathcal{V} \otimes \mathcal{I})^{\otimes n}(\Phi_n) \dd \nu_{\kappa}(V) \right\} \right) - \chi\left( \left\{   (\mathcal{U} \otimes \mathcal{I})^{\otimes n}(\Phi_n) \dd \mu(U) \right\} \right)  \right| \\
    =~&\left| S \left( \ee_{V \sim \nu_{\kappa} } \left[ (\mathcal{V} \otimes \mathcal{I})^{\otimes n}(\Phi_n) \right]  \right) - S\left( \int_{\U(d)} (\mathcal{U} \otimes \mathcal{I})^{\otimes n}(\Phi_n) \dd \mu(U) \right)  \right| \\
    \leq~&\frac{\log d_n}{2} \left\|\ee_{V \sim \nu_{\kappa} } \left[ (\mathcal{V} \otimes \mathcal{I})^{\otimes n}(\Phi_n) \right] - \int_{\U(d)} (\mathcal{U} \otimes \mathcal{I})^{\otimes n}(\Phi_n) \dd \mu(U)  \right\|_1 + \log 2 \\
    =~&\frac{\log d_n}{2} \left\| \left( \mathcal{Q}_{\nu_{\kappa}}^{(n)} - \mathcal{Q}_{\haar}^{(n)} \right)(\Phi_n)  \right\|_1 + \log 2 \\
    \leq~&\frac{\log d_n}{2} \left\| \mathcal{M}_{\nu_{\kappa}}^{(n)} - \mathcal{M}_{\haar}^{(n)} \right\|_{\diamond} + \log 2 \\
    \leq~& \frac{\kappa}{2} \log d_n + \log 2.
    \end{aligned}
    \end{equation}
    Combining Equation~\ref{eqn:HolevoInequality_1}, \ref{eqn:HolevoInequality_2} and \ref{eqn:HolevoInequality_3}, we have 
    $$
    \begin{aligned}
    \chi\left( \left\{   \psi_{P, V}  \dd \nu_{\kappa}(V) \right\} \right) &\geq \chi\left( \left\{ \mathcal{P}_V(\Phi_n) \dd \nu_{\kappa}(V) \right\} \right) \\
    &\geq \chi\left( \left\{   (\mathcal{V} \otimes \mathcal{I})^{\otimes n}(\Phi_n) \dd \nu_{\kappa}(V) \right\} \right) - 4 n \sqrt{2\epsilon} \log d_n - 2 \log 2 \\
    &\geq \chi\left( \left\{   (\mathcal{U} \otimes \mathcal{I})^{\otimes n}(\Phi_n) \dd \mu(U) \right\} \right) - \left( 4 n \sqrt{2\epsilon} + \frac{\kappa}{2} \right) \log d_n - 3 \log 2 \\
    &= S\left( \int_{\U(d)} (\mathcal{U} \otimes \mathcal{I})^{\otimes n}(\Phi_n) \dd \mu(U) \right) - \left( 4 n \sqrt{2\epsilon} + \frac{\kappa}{2} \right) \log d_n - 3 \log 2.
    \end{aligned}
    $$
    If we take $\ket{\Phi_n}$ as the maximally entangled state
    $$
   \ket{\Phi_n} = \bigoplus_{ \lambda \in \mathsf{S}_{n}^d } \frac{\dim W_{\lambda}^d}{\sqrt{d_n}} \ket{\Phi_{W_{\lambda}^d}^+} \otimes \ket{\eta_{\lambda}},
    $$
    where $\ket{\eta_{\lambda}}$ is again an arbitrary bipartite state on $V_{\lambda}^{\otimes 2}$. Using \cite[Lemma 3.3]{MetgerSimpleConstruction2024}, the Haar random moment operator turns $\Phi_n$ into
    $$
    \mathcal{Q}_{\haar}^{(n)}(\Phi_n) = \bigoplus_{\lambda \in \mathsf{S}_n^d } \frac{I_{{W_{\lambda}^d}^{\otimes 2}}}{d_n} \otimes \ket{\eta_{\lambda}} \bra{\eta_{\lambda}}.
    $$
    Therefore, $S\left(\mathcal{Q}_{\haar}^{(n)}(\Phi_n)\right) = \log d_n$. Putting everything together, and using the upper bound [cf. Lemma~\ref{lemma:HolevoInformation}] $\chi\left( \left\{   \psi_{P, V}  \dd \nu_{\kappa}(V) \right\} \right) \leq \log d_P$, we arrive at the following bound for the cost:
    \begin{equation}
    \label{eqn:DimensionLowerbound}
   \log d_P \geq \left(1 - 4 n \sqrt{2\epsilon} - \frac{\kappa}{2} \right) \log d_n - 3 \log 2. 
    \end{equation}
    The explicit expression for $d_n$ \cite{Regev1981} is given by
    $$
    d_n = \binom{n + d^2 - 1}{d^2 - 1}.
    $$
    Using Lemma~\ref{lemma:BinomialLowerbound}, we have 
    $$
    \begin{aligned}
        d_n \geq \frac{1}{n+ d^2} \left(1 + \frac{n}{d^2}  \right)^{d^2} \left(1 + \frac{d^2 - 1}{n+1} \right)^{n+1} \geq \frac{d^2}{(n + d^2)^2}\left(1 + \frac{n}{d^2}  \right)^{d^2} \left(1 + \frac{d^2}{n} \right)^{n},
    \end{aligned}
    $$
    having used $\left(1 + \frac{d^2 - 1}{n+1} \right)^{n+1} \geq \left(1 + \frac{d^2 - 1}{n} \right)^{n} = \left( 1 + \frac{d^2}{n} \right)^n \left(1 - \frac{1}{n+d^2}\right)^n \geq \frac{d^2}{n + d^2} \left( 1 + \frac{d^2}{n} \right)^n$. Moreover, using the Taylor series, we have 
    $$
    \begin{aligned}
    \left( 1 + \frac{n}{d^2}  \right)^{d^2} &= \exp\left\{ d^2 \log\left(1 + \frac{n}{d^2} \right)  \right\} \geq \exp\left\{ n -  \frac{n^2}{2d^2} \right\}, \\
    \left( 1 + \frac{d^2}{n} \right)^n &\geq \left( \frac{d^2}{n} \right)^n.
    \end{aligned}
    $$
    Since $n = \mathcal{O}(\polylog N)$ and $d = 2^N$, the quotient $n / d$ is vanishing for sufficiently large $N$. When $d \geq n$, it holds that
    $$
    d_n \geq \left( \frac{1}{d +\frac{n}{d} } \right)^2 \exp\left( -\frac{n^2}{2d^2} \right) \left(\frac{ed^2}{n} \right)^n \geq \frac{e^{-\frac{1}{2}}}{4d^2}  \left(\frac{ed^2}{n} \right)^n  \geq \frac{e^{-\frac{1}{2} }}{4}\left( \frac{ed}{n - 1}  \right)^n.
    $$
    Recall that we have set $n = \lceil \frac{(1- \kappa/2)\varpi}{4\sqrt{2\epsilon}} \rceil$, and thus 
    $$
    1 - 4n\sqrt{2\epsilon} - \frac{\kappa}{2} \geq 1 - 4\sqrt{2\epsilon} \left( \frac{(1- \kappa/2)\varpi}{4\sqrt{2\epsilon}} +1 \right) - \frac{\kappa}{2} = \left( 1 - \frac{\kappa}{2} \right)\left(1 - \varpi - 4\sqrt{2\epsilon}\right).
    $$
    Using the above inequality and Equation~\ref{eqn:DimensionLowerbound},
    $$
    \begin{aligned}
    \log d_P &\geq  {\left(1 - 4 n \sqrt{2\epsilon} - \kappa / 2 \right)n} \log \left( \frac{ed}{n - 1} \right) - \left(5 \log 2 + \frac{1}{2} \right) \\
    &\geq  { \varpi \left( 1 - \kappa / 2 \right)^2\left(\frac{1 - \varpi}{4\sqrt{2\epsilon}} - 1\right) } \log \left( \frac{ 4e\sqrt{2\epsilon} d }{(1 - \kappa/2) \varpi } \right) - \left(5 \log 2 + \frac{1}{2} \right),
    \end{aligned}   
    $$
    the desired result follows immediately from the relation $c_P = \log_2 d_P = \log d_P / \log 2$.
\end{proof}

\begin{remark}
    The memory size lower bound presented in \autoref{thm:DimensionLowerbound} has asymptotic $\log_2 d_P = \Omega\left(N \polylog N \right)$ when the parameters $\varpi$ and $\kappa$ are fixed as constants, $D \sim \polylog N$ and the programming error $\epsilon \sim 1/ \polylog N$, which improves the bound $c_P = \Omega\left(N  \log \log N \right)$ stated in \cite[Theorem 7]{Yang_2025} for storing the quantum states prepared by low-depth brickwork quantum circuits.
\end{remark}

Note that our result does not take the conventional form of the No-Programming Theorem, as is announced in \cite{MBQComp2007, Nielsen2012, Kubicki2019, YuxiangOptimal2020}. Instead of fixing the dimension of the unitaries and setting an infinitesimal programming error, we herein assume the error is bounded from below in asymptotics, yet is allowed to vanish as the scale of the circuit grows, or equivalently, $N \to \infty$. This somewhat resembles the notion of ``faithful'' quantum operations defined in \cite{CompressionIdenticalStates,Yang_2025}.

\begin{observation}[High randomness implies large program cost]
\label{observation:randomness_implies_program_cost}
    The main proof idea of \autoref{thm:DimensionLowerbound} is that we can quantify the Holevo information (and thus the program cost) of a low-depth circuit unitary ensemble via its capability to approximate Haar randomness (up to a circuit size-dependent moment $n$). We observe that a family of unitaries is hard to program if it scrambles the information well, that is, generates sufficient randomness. From the perspective of randomness generation, we can informally recover the no-programming theorem in a new way. Take a minimal-size exact unitary $t$-design ensemble $\{ U_x \dd x \}_{x \in \mathcal{X}}$, and the corresponding $s$-moment Choi state ensemble $\{ U_x^{\otimes s} \ket{\Phi_{s}} \bra{\Phi_s} {U_x^\dagger}^{\otimes s} \dd x\}_{x \in \mathcal{X}}$, where $\ket{\Phi_s} = \frac{1}{\sqrt{d^2}} \Ket{I^{\otimes s}}$. When $d^2 \lesssim s \leq t$, its Holevo information reads
    \begin{equation}
    \label{eqn:program_dimension_lowerbound}
    \begin{aligned}
    \log_2 d_P &\geq \chi\left( \left\{ U_x^{\otimes s} \ket{\Phi_s} \bra{\Phi_s} {U_x^\dagger}^{\otimes s} \dd x\right\}_{x \in \mathcal{X}} \right)
       \\
       &= S\left( \int_{\mathcal{X}} U_x^{\otimes s} \ket{\Phi_s} \bra{\Phi_s} {U_x^\dagger}^{\otimes s} \dd x \right) - \int_{\mathcal{X}} S \left( U_x^{\otimes s} \ket{\Phi_s} \bra{\Phi_s} {U_x^\dagger}^{\otimes s} \right) \dd x \\
        &= \sum_{\lambda \vdash_{d} s } \frac{\dim W_{\lambda}^d \dim V_{\lambda}}{d^s} \cdot \log \left( d^s \cdot \frac{\dim W_{\lambda}^d}{\dim V_{\lambda}} \right) \\
        &\geq - \log \left(  \frac{1}{d^{2s}} \sum_{\lambda \vdash_d s} \left(\dim V_{\lambda}\right)^2 \right)=  2s \log d - \log \left( \sum_{\lambda \vdash_d s} \left(\dim V_{\lambda}\right)^2 \right) \\
        &\gtrsim \frac{d^2-1}{2} \log(2s) + \frac{d-1}{2} \log(2\pi) - \frac{d^2}{2} \log d - \sum_{j=1}^d \log (j - 1)! = \Omega(d^2 \log s),
    \end{aligned}
    \end{equation}
    where the first inequality is due to the concavity of the logarithm, and the last asymptotic is due to Regev \cite[Corollary 4.4]{Regev1981} in the large $s$ regime. When $s > t$, the Holevo information yields a simple upper bound: Take any $t$-design $\mathsf{X}_t^d$, by minimality,
    $$
    \chi\left( \left\{  U_x^{\otimes s} \ket{\Phi_s} \bra{\Phi_s} {U_x^\dagger}^{\otimes s} \dd x\right\}_{x \in \mathcal{X}} \right) \leq \log \left| \left\{ U_x^{\otimes s} \ket{\Phi_s} \bra{\Phi_s} {U_x^\dagger}^{\otimes s} \dd x\right\}_{x \in \mathcal{X}} \right| \leq \log |\mathsf{X}_t^d|.
    $$
    By \cite[Theorem 21]{Roy_2009}, there exists a unitary $t$-design on $\U(d)$ with size $|\mathsf{X}_t^d| \leq \binom{d^2  + t - 1}{t}^2$. Therefore, when $s > t$, by minimality of the ensemble,
    $$
    \chi\left( \left\{  U_x^{\otimes s} \ket{\Phi_s} \bra{\Phi_s} {U_x^\dagger}^{\otimes s}\right\}_{x \in \mathcal{X}} \right) \leq \log \binom{d^2  + t - 1}{t}^2 = \mathcal{O}\left( d^2 \log t \right).
    $$
    This formulation reveals that the information capacity of a state ensemble carrying $s$ copies of a unitary $t$-design grows with respect to $s$ when $s \leq t$, and ceases to grow when $s$ exceeds $t$. Specifically, the amortized Holevo information\footnote{
    The amortization is to exclude the effect of the growth of system size.
    } $\frac{1}{\log s} \chi( \{ U_x^{\otimes s} \ket{\Phi_s} \bra{\Phi_s} {U_x^\dagger}^{\otimes s} \dd x\}_{x \in \mathcal{X}}) $ is non-vanishing for all $s \leq t$. Notably, requiring the ensemble of unitaries to form an exact $t$-design for every $t \in \mathbb{N}$ allows taking $s \to \infty$ in Equation~\ref{eqn:program_dimension_lowerbound}, which yields a request for infinite-size memory. Since $\U(d)$ is the only ensemble that satisfies the design requirement, this recovers the No-Programming Theorem.
\end{observation}

\subsection{Upper bound for the program cost}
\label{subsec:CostUpperbound}

A generic upper bound can be derived easily via a metric-space covering-net argument. The idea is intuitive: If the quantum memory can store the discretizations of the space of all unitary transformations with an appropriate resolution, then a carefully designed retrieval algorithm could program any unitary approximately. We formalize the terminology below.

\begin{definition}
\label{def:MetricSpace}
    Let $(\mathscr{X}, \mathfrak{d})$ be a metric space, let $\mathscr{K} \subseteq \mathscr{X}$, and fix $\epsilon > 0$. A subset $\mathscr{N}(\mathscr{K}, \mathfrak{d}, \epsilon) \subseteq \mathscr{K}$ is called an $\epsilon$-net of $\mathscr{K}$ if, for every $x \in \mathscr{K}$, there exists $y \in \mathscr{N}(\mathscr{K}, \mathfrak{d}, \epsilon)$ such that $\mathfrak{d}(x,y) \leq \epsilon$; its cardinality, $\lvert \mathscr{N}(\mathscr{K}, \mathfrak{d}, \epsilon) \rvert$, is called the $\epsilon$-covering number of $\mathscr{K}$.

\end{definition}

\begin{lemma}[{\cite[Lemma 9, reformulated]{Zhao_2024}}]
\label{lemma:CoveringNumberOfUnitaryGroup}
    For any error $0 < \epsilon \leq 1$, the $\epsilon$-covering number of the $d$-dimensional unitary group $\U(d)$ with respect to the diamond norm (distance) $\|\cdot\|_{\diamond}$ of the induced unitary channel satisfies the following inequality:
    $$
    \left|\mathscr{N} \left( \U(d), \|\cdot\|_{\diamond}, \epsilon  \right) \right| \leq \left(\frac{12}{\epsilon}\right)^{2d^2}.   
    $$
\end{lemma}
To derive an upper bound for the program cost, we first build an $\epsilon$-net for the unitaries that can be generated by a low-depth quantum circuit. Using a slightly modified version of Fact~\ref{fact:errorPropagation}, we can arrive at the following statement.

\begin{lemma}
\label{lemma:EpsilonNetForBrickworkCircuit}
    Let $\mathscr{U}_{k, \ell, D} \subseteq \U(2^N)$ be the set of $N$-qubit unitaries that can be generated by a brickwork circuit on $\ell$ $k$-qubit unitaries with depth $D$. For any $0 < \epsilon \leq 1$, the $\epsilon$-covering number of unitary channels that corresponds to unitaries in $\mathscr{U}_{k, \ell, D}$ in diamond norm distance satisfies
    $$
    \left|\mathscr{N}(\mathscr{U}_{k, \ell, D}, \|\cdot\|_{\diamond}, \epsilon)\right| \leq \left[   \left( \frac{eN}{k} \right)^k \left( \frac{12\ell}{\epsilon} \right)^{2^{2k+1}} \right]^{\ell}.
    $$
\end{lemma}

\begin{proof}
    Recall Definition~\ref{def:brickwork_circuit}. The set of unitaries in $\mathscr{U}_{k, \ell, D}$ can be written explicitly as
    $$
    \mathscr{U}_{k, \ell, D} = \left\{ \prod_{r=1}^{D} \prod_{j \in \mathcal{L}_r } G_j^{q_j}: \quad G_j \in \U(2^k),~\sum_{r=1}^D |\mathcal{L}_r| = \ell,~|q_j| = k  \right\},
    $$
    where we use the superscript $q_j$ to indicate that $G_j$ acts on the qubits in $q_j$ and suppress the identity on other idle qubits. The proof idea resembles that of \cite[Theorem 8]{Zhao_2024}, as we can rewrite the product operation $\prod_{r=1}^{D} \prod_{j \in \mathcal{L}_r } G_j^{q_j} = \prod_{j=1}^{\ell} G_j^{q_j}$ with a proper indexing of the gates. Using the error propagation relation in Fact~\ref{fact:errorPropagation} together with the union bound. Let $\mathscr{Q}(k, C) = \left\{  
     q \subseteq [N]:\, |q| = k,C[q] \text{ is connected}
     \right\}$ denote the valid $k$-qubit supports on which the $k$-local gates may act. When the circuit is all-to-all (\ie, the connectivity graph $C = K_N$, the $N$-complete graph), there are $\binom{N}{k}$ pairs of qubits that each local unitary gate can act on. Therefore, the cardinality of the covering net $\mathscr{N}(\mathscr{U}_{k, \ell, D}, \|\cdot\|, \epsilon)$ yields an upper bound
    $$
    \begin{aligned}
    \left|\mathscr{N}(\mathscr{U}_{k, \ell, D}, \|\cdot\|_{\diamond}, \epsilon)\right| &\leq \left[ \left|\mathscr{Q}(k, C) \times  \mathscr{N}(\U(2^k), \|\cdot\|_{\diamond}, \epsilon / \ell)\right| \right]^{\ell} \\
    &\leq \left[ |\mathscr{Q}(k, K_N) |\left|\mathscr{N}(\U(2^k), \|\cdot\|_{\diamond}, \epsilon / \ell)\right| \right]^{\ell} \\
    &= \left[ \binom{N}{k} \left|\mathscr{N}(\U(2^k), \|\cdot\|_{\diamond}, \epsilon / \ell)\right| \right]^{\ell}.
    \end{aligned}
    $$
    The result follows immediately via the handy inequality $\binom{N}{k} \leq \left(\frac{eN}{k}\right)^k$ and Lemma~\ref{lemma:CoveringNumberOfUnitaryGroup}.
\end{proof}
We can readily obtain the following upper bound for the program cost.
\begin{theorem}
\label{thm:DimensionUpperbound}
    Programming a brickwork quantum circuit on $N$ qubits with depth $D$ and $\ell$ $k$-local gates up to error $\epsilon \in (0, 1]$ in diamond norm distance requires a quantum processor with program cost $c_P$ that satisfies
    \begin{equation}
    \label{eqn:ProgramCost}
    c_P \leq k \ell \log_2 \left(\frac{eN}{k} \right) + 2^{2k+1} \ell \log_2 \left(\frac{12\ell}{\epsilon}\right) .
    \end{equation}
\end{theorem}

\begin{proof}
     By definition, it suffices to construct an $\epsilon$-universal processor $(\mathcal{C}, \{\psi_{P, U}\}_{U \in \mathscr{U}_{k, \ell, D}})$ that can coherently program the brickwork circuit unitaries from the $\epsilon$-net $\mathscr{N}(\mathscr{U}_{k, \ell, D}, \|\cdot\|_{\diamond}, \epsilon) = \left\{ U_1, \dots, U_{M} \right\}$. The processor can be explicitly constructed via postselection:
     \begin{equation}
     \label{eqn:ProcessorConstruction}
     \begin{aligned}
         \forall \, \rho \in \Dens(\ml{H}),~\mathcal{C}(\rho \otimes \psi_{P, U}) &= \sum_{j=1}^M \bra{j} \psi_{P, U} \ket{j} \cdot \mathcal{U}_j(\rho), \\
         \quad \psi_{P, U} &= \ket{t} \bra{t}:~\left\| \mathcal{U}_t - \mathcal{U} \right\|_{\diamond} \leq \epsilon.
     \end{aligned}
     \end{equation}
     The set of desired programming states $\{\psi_{P, U}\}_{U \in \mathscr{U}_{k, \ell, D}}$ exists due to the definition of the $\epsilon$-net. Finally, in the sense of $\epsilon$-approximate programmability, the program dimension for the above scheme can be bounded from above by
     $$
     \begin{aligned}
         d_P &= \dim \left( \Span\left\{ \psi_{P, U}~|~U \in \mathscr{N}(\mathscr{U}_{k, \ell, D}, \|\cdot\|_{\diamond}, \epsilon) \right\}  \right) \\
         &\leq \left|\mathscr{N}(\mathscr{U}_{k, \ell, D}, \|\cdot\|_{\diamond}, \epsilon)\right|.
     \end{aligned}
     $$
     Taking the logarithm on both sides of the inequality, the proof is completed with the upper bound presented in Lemma~\ref{lemma:EpsilonNetForBrickworkCircuit}.
\end{proof}

\begin{remark}
    With the constraint $\ell \leq \frac{ND}{k}$, if we set the asymptotics of the parameters identical to those in \autoref{thm:DimensionLowerbound}, then $c_P \leq ND \log_2\left(\frac{eN}{k}\right) + \frac{2^{2k+1}}{k} ND \log_2 \left(\frac{12ND}{k\epsilon}\right) = \mathcal{O}\left(N \polylog N\right)$. Combining the lower bound presented in \autoref{thm:DimensionLowerbound}, we obtain a tight characterization of the asymptotic program cost for programming unitaries in $\mathscr{U}_{k, \ell ,D}$ when $D \sim \polylog N$, that is, $c_P = \Theta\left(N \polylog N \right)$. Although learning a low-depth brickwork circuit in $\mathscr{U}_{k, \ell, D}$ can be exponentially hard in the worst case \cite[Theorem 3]{HuangLearnShallowCircuit2024}, programming it can be efficient with prior knowledge of the circuit architecture.
\end{remark}

\subsection{Trade-off between circuit architecture complexity and local gate program cost}
\label{sec:TradeOff}

Although the upper bound presented in Section~\ref{subsec:CostUpperbound} is a coarse estimation, it suggests the following fact: The program state of a unitary generated by a quantum circuit encodes information not only of the parameters of the local unitary gates, but also how the gates are applied among the registers legitimately [cf. Definition~\ref{def:brickwork_circuit}]. 

Assume that the location $(r: j \in \mathcal{L}_r;~q_j)$ about each unitary gate $G_j$ is encoded in a bit-string $\mathsf{L}_j \in \{0, 1\}^{m}$, where $m$ is a constant dependent on the circuit architecture\footnote{
For instance, we can take $m$ such that $2^m$ is greater than or equal to the number of valid pairs of qubits on any layer of the circuit, where the $k$-local gates can be applied. Note that the argument applies to quantum brickwork circuits with any geometry.
}, and the sender (the user who prepares the program state) and receiver (the quantum processor) have reached a consensus on the circuit architecture (the connectivity graph, circuit depth, gate count, \etc). At such, a candidate of the program state is given by the following tensor product state:
$$  
\psi_{P, U} = \bigotimes_{j=1}^{\ell} \left(\ket{\mathsf{L}_j} \bra{\mathsf{L}_j} \otimes \psi_{P, G_j} \right).
$$
While the state does not yield the most compact form, its expression reveals an underlying regularity: When the local gates become larger, the program cost of each gate increases, while there are fewer gates allowed to be placed in the circuit; thus, both the gate number $\ell$ and the bit-string length $|\mathsf{L}_j|$ decrease. This highlights a trade-off in program cost between describing the circuit architecture and storing individual program states of local gates.

\begin{figure}[t!]
  \centering
  \subfloat[Primitive 1D circuit]
  {\includegraphics[width = 0.47\textwidth]{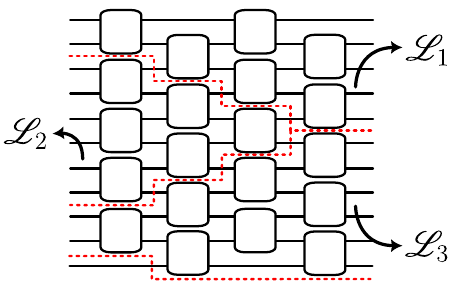}\label{fig:1Dbefore}} \hfill
  \subfloat[Reduced 1D circuit] 
  {\includegraphics[width = 0.47\textwidth]{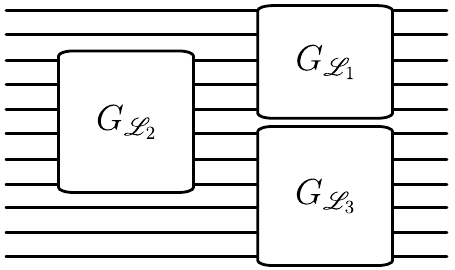}\label{fig:1Dafter}} \hfill 
  \subfloat[Primitive 2D circuit]
  {\includegraphics[width = 0.47\textwidth]{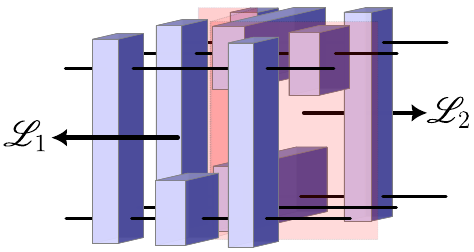}\label{fig:2Dbefore}} \hfill 
  \subfloat[Reduced 2D circuit]
  {\includegraphics[width = 0.47\textwidth]{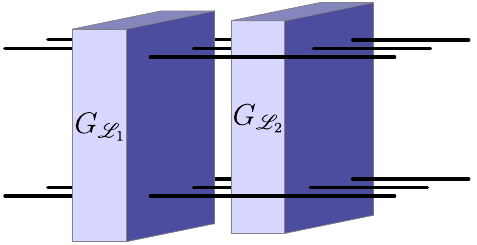}\label{fig:2Dafter}} \hfill
  \caption{Examples of light-cone reduction of 1D and 2D 2-local brickwork circuits.}
  \label{fig:LightconeReduction}
\end{figure}

\par How can this flexibility be utilized, given that the gates have fixed dimensions? We invoke the well-applied light-cone argument for general quantum circuits \cite{Haferkamp2022, HuangLearnShallowCircuit2024, Nadimpalli_2024, Yang_2025}, which indicates that the local gates can be grouped into non-intersecting light-cones without violating their relative order of implementation. Each group of small unitaries is combined into a larger unitary, while the layout of these resulting unitaries can be significantly simpler. To specify, the light-cones are defined by their first layer of gates. Starting from the first layer, each subsequent layer is constructed inductively by including the set of gates that act on the qubits affected by the gates in the previous layer. The remaining separated blocks between the forward light-cones are also grouped to form the backward light-cones.
A schematic illustration of this reduction is provided in \autoref{fig:LightconeReduction}. 

\par One would naturally ask whether the reduction help reduce the program cost of brickwork circuits. To address this, we analyze the cost bounds in two scenarios: when the local unitaries are either generic or have a specific structure.

\subsection{Reduction of circuits with general local unitaries}
\label{subsection:LightConeReductionGeneralUnitary}

Assume that each light-cone yields depth $W$, and we compose the gates in light-cone $\mathscr{L}$ to form the light-cone gate $G_{\mathscr{L}}$. The circuit is thus reduced to $\lceil\frac{D}{W}\rceil$ layers. Suppose that the gates are interlaced so that the width of each light cone (\ie, the number of qubits it occupies) scales as \(\Theta(W)\). After reduction, the circuit contains $\frac{D}{W} \cdot \Theta(\frac{N}{W}) = \Theta(\frac{ND}{W^2})$ gates. In comparison, the primitive circuit contains $\Theta(ND)$ gates, each being $\mathcal{O}(1)$-local. Let $c_P$ and $c_P^r$ denote the $\epsilon$-universal program costs of the primitive and reduced all-to-all circuits, respectively. Applying \autoref{thm:DimensionUpperbound} gives
$$
\begin{aligned}
c_P &\lesssim  ND \log_2 N + ND \log_2 \left( \frac{ND}{\epsilon} \right), \\
c_P^r &\lesssim \frac{ND}{W} \log_2 \left( \frac{N}{W} \right) + 2^{\Theta(W)} \frac{ND}{W^2} \log_2 \left( \frac{ND}{W^2\epsilon} \right).
\end{aligned}
$$
In the large $N$ regime, the second terms in the bounds above dominate asymptotically. To ensure that the reduced quantum circuit is less costly, \ie, $c_P^r = o( c_P)$, one forces
$$
   \frac{2^{\Theta(W)}}{W^2} \log_2 \left( \frac{ND}{W^2 \epsilon} \right) = o \left(\log_2\left( \frac{ND}{\epsilon} \right) \right) \quad \implies \quad \epsilon = \omega\left( \frac{ND}{ W^{ \frac{2}{1 - W^2 / 2^{\Theta(W)} } } } \right).
$$
To validate the bound with at most constant error, we require that $W^{ \frac{2}{1 - W^2 / 2^{\Theta(W)} } } = \omega( ND)$. Since the quotient $W^2 / 2^{\Theta(W)}$ is vanishing, for sufficiently large $W$, $W^{ \frac{2}{1 - W^2 / 2^{\Theta(W)} } } \leq W^{2 + \varsigma}$ for some $\varsigma < 1$. Thus, the previous condition gives $W = \omega\left( (ND)^{\frac{1}{2 + \varsigma}}\right)$. Under the low-depth circuit assumption [cf. Definition~\ref{def:brickwork_circuit}], $N = \omega( D^{z})$ for any $z \in \mathbb{N}$. Substitute this into the lower bound, $W = \omega( D^{\frac{z + 1}{2 + \varsigma}})$, contradicting the fundamental constraint $W \leq D$ when $z \geq 2$ even in the shallow-circuit setting $D = \mathcal{O}(1)$ \cite{HuangLearnShallowCircuit2024}. 
\par We hereby conclude that the reduction does not save the asymptotic cost of storing the program state in the general case\footnote{
The light-cone argument is also invalidated in bounding the computational power of $\polylog N$-depth quantum circuits \cite{Nadimpalli_2024, QAC0SuperLinearAncillae}.
}. This is not a surprising result, as unitaries are continuous objects with an exponentially large number of free parameters. In contrast, the layout of the quantum gates within the circuit is a discrete object that yields simple and compact descriptions. In more general cases, the width of the light-cones can grow exponentially in their depth \cite[Definition 10]{HuangLearnShallowCircuit2024}, making it even more costly to program the light-cone unitaries. Furthermore, the analysis above ensures that the bound presented in \autoref{thm:DimensionUpperbound} is majorly optimal, without any prior assumption on the local unitaries.

\subsection{
Light-cone argument reduces the cost for structured local unitaries
}
Although the discussion provided in Section~\ref{subsection:LightConeReductionGeneralUnitary} shows that for programming brickwork circuits with generic local unitaries, the light-cone argument does not reduce program cost for us, there are certainly examples where it works. By relaxing the requirement of universality and restricting attention to a specific family of unitary operators that admit an efficient parametrization, the light-cone unitary can be described using a set of parameters whose growth rate is lower than that of the parameters defining the local unitaries it contains. Mathematically, the set $\mathsf{P}_{k} \subseteq \mathbb{R}^{2^{2k}}$ depicts the free real parameters of those $k$-local unitaries, and similarly $\mathsf{P}_{\mathscr{L}}$ depicts that of a single light-cone unitary $G_{\mathscr{L}}$\footnote{
Note that this notion is ad hoc, but we can instead treat $\mathsf{P}_{\mathscr{L}}$ as a subset of $\mathsf{P}_k^{|\mathscr{L}|}$.
}. Denote the program cost of the local and light-cone unitaries as $c_P(\mathsf{P}_k)$ and $c_P(\mathsf{P}_{\mathscr{L}})$ respectively. Suppose a brickwork circuit on $\ell$ gates is reduced to light-cones $\left\{ \mathscr{L}_1, \dots, \mathscr{L}_h \right\}$. The asymptotic program cost of the circuit is reduced when the following condition is satisfied:
\begin{equation}
\label{eqn:CostReductionCondition}
\sum_{j=1}^h c_P(\mathsf{P}_{\mathscr{L}_j}) = o \left( \ell \cdot c_P(\mathsf{P}_k)  \right).
\end{equation}
For intuitiveness, we provide a concrete example\footnote{
We believe that these examples are sparse within the universe of parameterized unitaries.
} herein.

\begin{example}
\label{example}
\label{example:program_cost_reduction_example}
    Define the following set of unitaries on $k$ qubits:
    $$
    \mathscr{U}_{k, P} = \left\{ 
    e^{\mathrm{i} \theta \prod_{j=1}^{k} P^{q_j} }~\bigg|~ \theta \in [0, 2\pi),~ q \subset [N],~ |q| = k
    \right\}, \quad P \in \{X, Y, Z\}.
    $$
    We restrict the brickwork circuit of interest to be chosen from the set $\langle\mathscr{U}_{k, P}\rangle$, \ie, $\mathscr{U}_{k, P}$ serves as its generating set, and assume $C = K_N$. Still, the program cost originates from the uncertainty about the local gate type and which set of qubits it applies to. Thus, $\mathsf{P}_k = [0, 2\pi) \times \mathscr{Q}(k, K_N)$. 
    It can be readily verified that $[\prod_{j=1}^k P^{q_j}, \prod_{j=1}^{k} P^{q_j'}] = 0$ for any set of indices $q, q'$. Moreover, we assume that in the light-cone there are $T_{\mathscr{L}} = o(m_{\mathscr{L}})$ distinct values of $q$, expressed as $\{q_1, \dots, q_T\} \subseteq \mathscr{Q}(k, K_N)$. Therefore, if $\mathscr{L} = \{ G_1, \dots, G_{m_{\mathscr{L}}}\}$ where $G_t = e^{\mathrm{i} \theta_t \prod_{j=1}^{k} P^{q_{t, j}} }$, the light-cone gate can be expressed as
    $$
    G_{\mathscr{L}} = e^{\mathrm{i} \left( \sum_{t = 1}^{m_{\mathscr{L}}} \theta_t \prod_{j=1}^{k} P^{q_{t, j}}  \right)
    }.
    $$
    To perform an $\epsilon$-approximate programming of each local unitary gate, we build an $\epsilon$-net $\mathscr{N}([0, 2\pi), |\cdot|, \epsilon)$ for $[0, 2\pi)$. For any $\theta \in [0, 2\pi)$, there exists an angle $\theta_{\epsilon} \in \mathscr{N}([0, 2\pi), |\cdot|, \epsilon)$ such that $|\theta - \theta_{\epsilon}| \leq \epsilon$. Similar to Equation~\ref{eqn:ProcessorConstruction}, the processor is constructed as $$
    \begin{aligned}
   \mathcal{C}(\rho \otimes \psi_{P, G_t})(\cdot) &= \sum_{ \hat \theta \in \mathscr{N}([0, 2\pi), |\cdot|, \epsilon)  } \sum_{ q \in \mathscr{Q}(k, C) } \bra{\hat \theta, q} \psi_{P, G_{t}} \ket{\hat \theta, q} \cdot e^{\mathrm{i} \hat \theta \prod_{j=1}^{k} P^{q_j} } (\cdot) e^{-\mathrm{i} \hat \theta \prod_{j=1}^{k} P^{q_j} }, \\
   \psi_{P, G_t} &= \ket{\widetilde{\theta}_t, q_t} \bra{\widetilde{\theta}_t, q_t}:~|\widetilde{\theta}_t - \theta_t| \leq \epsilon.
   \end{aligned}
    $$ Using {\cite[Example 5]{aleksandrov2016operatorlipschitzfunctionsenglish}}, it follows that
    \begin{equation}
    \label{eqn:PhaseGateError}
    \begin{aligned}
    \frac{1}{2}\left\|
    \mathcal{C}(\rho \otimes \psi_{P, G_t}) - \mathcal{G}_t
    \right\|_{\diamond} &\leq  \left\| e^{\mathrm{i} \widetilde{\theta}_t \prod_{j=1}^{k} P^{q_{t, j}} } - e^{\mathrm{i} \theta_t \prod_{j=1}^{k} P^{q_{t, j}} }  \right\| \\
    &\leq |\widetilde{\theta}_t - \theta_t| \left\| e^{\mathrm{i} \prod_{j=1}^{k} P^{q_{t, j}} } \right\| \\
    &\leq \epsilon.
    \end{aligned}
    \end{equation}
    For the light-cone gate $G_{\mathscr{L}}$, we can rewrite its expression by our assumption, by $$e^{\mathrm{i} \left( \sum_{t = 1}^{m_{\mathscr{L}}} \theta_t \prod_{j=1}^{k} P^{q_{t, j}}  \right)
    } = e^{\mathrm{i} \left( \sum_{r=1}^{T_{\mathscr{L}}} \left(\sum_{t: q_t = q_r} \theta_t \bmod 2\pi \right) \prod_{j=1}^{k} P^{q_{r, j}} \right)},
    $$, thereby reducing the number of free parameters. Assume that $G_{\mathscr{L}}$ acts on $k_{\mathscr{L}}$ qubits and depends on $T_{\mathscr{L}}$ angles. Accordingly, we can write $\mathsf{P}_{\mathscr{L}} = [0, 2\pi)^{T_{\mathcal{L}}} \times \mathscr{Q}(k_{\mathscr{L}}, N)$. Still, we assume that the circuit is sufficiently dense with $\ell \sim ND$, $m_{\mathscr{L}_j} \sim W^2$, $h \sim \frac{ND}{W^2}$, and $k_{\mathscr{L}_j} \sim W$ for any $j \in [h]$. To ensure that the circuit is $\epsilon$-universally programmed, the local unitaries are $\epsilon / \ell$-approximate, while the light-cone unitaries are $\epsilon/h$-approximate. Specifically, the free angles for local unitaries and the light-cone unitaries are to be approximated to error $\epsilon / \ell$ and $\epsilon / h T_{\mathscr{L}_j}$, in light of Equation~\ref{eqn:PhaseGateError}. Therefore, the primitive circuit and the reduced circuit can be bounded from above by
    $$
    \begin{aligned}
    \sum_{j=1}^{h} c_P(\mathsf{P}_{\mathscr{L}_j}) &\lesssim \sum_{j=1}^h T_{\mathscr{L}_j} \log_2  \left( \frac{2\pi hT_{\mathscr{L}_j}}{\epsilon} \right) + T_{\mathscr{L}_j} k_{\mathscr{L}_j} \log_2 \left(\frac{eN}{k_{\mathscr{L}_j}} \right), \\
    \ell \cdot c_P(\mathsf{P}_k) &\lesssim \ell \cdot \left(\log_2\left( \frac{2\pi\ell}{\epsilon} \right) + k \log_2 \left(\frac{eN}{k} \right)  \right).
    \end{aligned}
    $$
    One can readily verify that the condition presented in Equation~\ref{eqn:CostReductionCondition} is satisfied. 
    Therefore, we arrive at a processor that fully encodes not only the gate parameters but also the ``light-cone-reduced'' circuit architecture, operating at a lower overall program cost.
\end{example}

\subsection{Going beyond light-cone: The landscape of reducing the program cost for parameterized quantum circuits}
The failure to reduce program cost via the light-cone argument reveals a fundamental difference between the parameterization of a quantum circuit and that of the distribution it generates. The former grants us full knowledge of the circuit geometry, while the latter is merely a ``shadow'' of the unitary circuit\footnote{
That is, we are only concerned about the parameters revealed by a limited collection of measurements, as per the case in quantifying the capability of quantum circuits in solving decision problems.
} \cite{Kunjummen_2023}. Although the light-cone argument enables us to prove lower bounds on the capability of solving promise problems \cite{Bravyi2018, Nadimpalli_2024} and generating complex distributions \cite{Anshu_2023}, simplifying the circuit layout introduces an increase in complexity inside the light-cone unitaries, as demonstrated in Section~\ref{subsection:LightConeReductionGeneralUnitary}. Consequently, within the programming framework, this simplification does not fundamentally cause the parameter space to degenerate, thus maintaining the program cost.
\par More broadly, one may ask how to optimize more practically relevant NISQ circuit families, such as the Hardware-Efficient Ans\"atze (HEA) \cite{Leone2024practicalusefulness}. While the light-cone argument can be applied to HEA circuits with bounded gate locality (which resembles a $\mathsf{QNC}$ circuit), it does not reduce the parameter count \cite[Section 3.2]{Leone2024practicalusefulness}. For a generic family of parameterized quantum circuits, we present the following proposition, as an extension to Section~\ref{subsec:CostLowerbound} and \ref{subsec:CostUpperbound}. It suggests that finding an exact parameter space of it directly implies a lower and upper bound for the program cost for its $\epsilon$-approximate processor.

\begin{proposition}[Program cost $\approx$ parameter count]
\label{proposition:program_cost_from_parameter_count}
    Let \(\mathscr U_{\mathsf P}=\{U(\bm\theta)\}_{\bm\theta\in\mathsf P}\subseteq\U(d)\) be a unitary family parameterized by a compact subset \(\mathsf P\) of a real vector space \(T\), and suppose that \(\mathsf P\) is contained in a ball of radius \(R\) centered at the origin of \(T\). If the parameter-to-unitary map \(\pi:\mathsf P\to\mathscr U_{\mathsf P}\) is \(L\)-Lipschitz, then for all $0 < \delta < 1 - 4\sqrt{2 \epsilon}$, the program cost of $\epsilon$-approximate programming of the family $\mathscr{U}_{\mathsf{P}}$ satisfies
    $$
    \begin{aligned}
    \left(1 - \delta - 4  \sqrt{2 \epsilon}\right) (\dim T - 1) \log_2 &\left( 1 + \frac{1}{\dim T-1} \left\lceil \frac{\delta}{4 \sqrt{2\epsilon}} \right\rceil \right) - \mathcal{O}(1) \\
    & \leq \log_2 d_P = c_P \leq \dim T \cdot  \log_2 \left( 1 + \frac{\mathcal{O}(RL)}{\epsilon} \right).
    \end{aligned}
    $$
\end{proposition}

\begin{proof}
    The proof idea resembles that of \autoref{thm:DimensionLowerbound} and \autoref{thm:DimensionUpperbound}, with slight modifications. For the lower bound, applying the Schur-Weyl duality on the $n$-fold parameter space yields $T^{\otimes n} \cong \bigoplus_{\lambda \vdash_{\dim T} n } T_{\lambda} \otimes V_{\lambda}$. The dimension $d_n$ in the formulation of the Holevo information is replaced by 
    $$
   d_{n, T} = \sum_{\lambda \vdash_{\dim T} n} \left( \dim T_{\lambda} \right)^2 = \dim \left( \vee^n (T) \right) = \binom{n + \dim T - 1}{\dim T - 1},
    $$
    where $\vee^n T$ is the symmetric subspace of $T^{\otimes n}$. Then, by an analogous processor-reusing argument in \autoref{thm:DimensionLowerbound}, but without the randomness truncation due to circuit depth constraints, we can obtain a lower bound
    $$
    \begin{aligned}
    \log_2 d_P \geq \left( 1 - 4n\sqrt{2 \epsilon} \right) \log_2 d_{n, T} - \mathcal{O}(1) = \left( 1 - 4n\sqrt{2 \epsilon} \right) \log_2 \binom{n + \dim T - 1}{\dim T - 1} - \mathcal{O}(1).
    \end{aligned}
    $$
    Taking $n = \lceil \frac{\delta}{4\sqrt{2\epsilon}} \rceil$, and using the inequality $\binom{a + b}{b} \geq \left( 1 + \frac{a}{b} \right)^{b}$ concludes the proof for the lower bound. As for the upper bound, the Lipschitz condition yields for $\bm{\theta}, \bm{\theta}' \in \mathsf{P}$, 
    $$
   \left\| \mathcal{U}(\bm{\theta}) - \mathcal{U}(\bm{\theta}') \right\|_{\diamond} \leq 2\left\| U(\bm{\theta}) - U(\bm{\theta}') \right\| \leq 2L \left\| \bm{\theta} - \bm{\theta}' \right\|_2.
    $$
    Therefore, by Fact~\ref{fact:errorPropagation}, an $\epsilon'=\epsilon/(2L)$ covering net of $\mathsf{P}$ in the $\ell_2$ norm suffices to achieve $\epsilon$-approximate programmability. By a standard volumetric argument \cite{vershynin2011introductionnonasymptoticanalysisrandom}, if we denote the leading constant $C_x = \frac{\pi^{x/2}}{\Gamma(x/2 + 1)}$, and use $\mathcal{B}(r)$ to denote the radius-$r$ ball centered at the origin, it holds that
    $$
    \begin{aligned}
    &\left|\mathcal{N}(\mathsf{P}, \left\|\cdot\right\|_2, \epsilon')\right| \cdot \textbf{Vol}\left(\mathcal{B}(\epsilon' / 2) \right) \leq \textbf{Vol}\left( \mathcal{B}(R + \epsilon' / 2) \right) \\
    &\quad \implies \left|\mathcal{N}(\mathsf{P}, \left\|\cdot\right\|_2, \epsilon')\right| \leq \frac{C_{\dim T}\left( R + \epsilon' / 2 \right)^{\dim T}}{C_{\dim T}\left( \epsilon' / 2 \right)^{\dim T}} = \left( 1 + \frac{2R}{\epsilon'} \right)^{\dim T} = \left( 1 + \frac{4RL}{\epsilon} \right)^{\dim T}.
    \end{aligned}
    $$
    The upper bound is obtained readily by taking the logarithm on both sides of the inequality, concluding the proof. As a sanity check, this is consistent with Observation~\ref{observation:randomness_implies_program_cost}: a unitary design that matches sufficiently high moments must have a high-dimensional parameter space in which to encode its ``random seed''.
\end{proof}
Proposition~\ref{proposition:program_cost_from_parameter_count} has established the resource equivalence between the program cost and the parameter count of unitary families. Notably, when $\dim T = d^2$ and $L, R = \mathcal{O}(1)$, it recovers the quantitative No-Programming Theorem in \cite{Kubicki2019, YuxiangOptimal2020}. The analysis of optimizing program cost for unitaries generated by quantum circuits of specific structures naturally reduces to optimizing the parameter count for them. However, this does not immediately apply to our analysis of the programmability of $\mathsf{QNC}$ circuits, as the minimal parameterization of a bounded-size quantum circuit is still open\footnote{We conjecture that solving this leads to solving several open problems in circuit lower bound problems, for instance, unconditional quantum advantage of $\mathsf{QNC}^0$ over circuits other than $\mathsf{NC}^0$ \cite{Bravyi2018}, \eg, $\mathsf{NC}^1$.}.
\par We remark that this is a difficult task, as demonstrated by several prior works. In \cite{chia_et_al:LIPIcs.ITCS.2022.47}, the authors define the \textsc{UMCSP} problem: given a complete description of a unitary $U$, decide whether there exists a small (in gate count) circuit that approximately implements $U$. Adapting to the programming task, if we are given access to the description of a generic unitary family $\mathscr{U}_{\mathcal{X}} =\{U_x\}_{x \in \mathcal{X}}$ and a \textsc{UMCSP} oracle, we can decide whether there exists a compact circuit family that approximately generates $\mathscr{U}_{\mathcal{X}}$, yielding a small parameter space, and thus small program cost. In \cite{vandewetering2024optimisingquantumcircuitsgenerally, Wetering_2025}, the authors investigate the hardness of deciding whether the parameters of a given parameterized quantum circuit can be fused to reduce the parameter count. The main result of \cite{chia_et_al:LIPIcs.ITCS.2022.47} shows that \textsc{UMCSP} with a small promise gap is in $\mathsf{QCMA}$ and, assuming the existence of computationally secure cryptography, is not in $\mathsf{BQP}$. Even given a prior guarantee on the circuit being a parameterized circuit, \cite{vandewetering2024optimisingquantumcircuitsgenerally, Wetering_2025} shows that deciding whether the program cost of a parameterized circuit family can be reduced is $\mathsf{NP}$-hard, if no further structural assumptions are made. In both cases, the search-to-decision reduction requires exponential queries to the decision oracle to synthesize an explicit programming scheme. This explains why constructing a general protocol for reducing program cost is intrinsically hard beyond specific examples such as that in Example~\ref{example:program_cost_reduction_example}.

\par Observing these hardness results, and the relations presented in Proposition~\ref{proposition:program_cost_from_parameter_count}, we raise the following question:

\begin{problem}
    What is the hardness of deciding whether the optimal program cost of $\epsilon$-approximate programming of a low-depth circuit unitary family $\{ U_x \}_{x \in \mathcal{X}}$ is beyond or below two gapped thresholds, under different choices of $\epsilon$? Can we synthesize the cost-optimal programming scheme via search-to-decision reduction by solving this decision problem?
\end{problem}

\section{Conclusion}
\subsection{Concluding this work}
In this work, we present a comprehensive analysis of the computational and storage resources required to program the unitaries generated by low-depth brickwork quantum circuits. Notably, the scaling of these requirements for programming generic unitaries is far from optimal in the low-depth regime. We find that programming $N$-qubit low-depth brickwork circuits can be efficient with respect to both circuit depth and program cost. When the circuit depth is polylogarithmic, information-theoretic arguments show that the program-cost scaling $\sim N\polylog N$ is tight for macroscopic $N$, corresponding to a large-scale quantum device. This result paves the way for programmable quantum computers capable of running NISQ algorithms. We further examine whether the conventional light-cone argument can reduce the upper bound for program costs, assuming generic and unstructured local unitaries, and find that gate-wise programming is essentially optimal in the low-depth regime while the light-cone grouping only helps for certain structured circuit families.

The above discovery allows us to rethink the discussion about the performance -- or resource-error trade-off -- equivalence programming $\approx$ metrology $\approx$ learning \cite{YuxiangOptimal2020}, from the setting of universal unitaries to that of low-depth circuit unitaries. By \cite[Theorem 18]{Zhao_2024}, there exists a unitary generated by a $\polylog N$-depth brickwork circuit that can not be learned efficiently unless $\mathsf{RingLWE}$ is polynomial-time solvable, even with non-vanishing error, while programming it is efficient. As suggested by \cite{LearningImpliesLowerbounds2025}, greater difficulty in synthesizing states and unitaries corresponds to greater difficulty in learning them. Informally speaking, the hardness of programming and learning unitaries coincide in the worst and universal case but largely separate when restricted to NISQ circuits. Therefore, the previous conjecture is likely to break down.

\subsection{Future perspectives}

The derived program-cost bounds may be extended to noisy quantum circuits using the Stinespring dilation theorem \cite{Nielsen2012}. However, restricting the resulting dilation to low-depth quantum channels and deriving a corresponding program-cost lower bound are highly nontrivial. For the upper bound, applying the covering/packing argument in our construction, or that in \cite{Kubicki2019}, would suffice. Meanwhile, we establish the lower bound via the randomness-program cost relation [cf. Observation~\ref{observation:randomness_implies_program_cost}], which is not only restricted to low-depth circuits but also applies to quantum circuits of arbitrary geometry. Our proof techniques can also be utilized to recover results in some prior works that have similar settings to unitary programming. For instance, the gate compression protocol \cite{Chiribella_2015} involves an input state holder (Alice) and a gate holder (Bob) who are to apply the gate on the state collaboratively with minimum communication\footnote{
The program cost is essentially a resource metric of quantum communication between the commander and the processor.
}, allowing failure on some randomly chosen state with negligible probability. Our techniques can be utilized to provide the upper and lower bounds for the communication cost, using concentration of the Haar measure under the Schur-Weyl basis (see, \eg, \cite{Chiribella_2015, Chiribella2016}).

Further research on NISQ-circuit programming could include \textbf{(1)} investigating whether succinctly encoding the structural information of the circuit into the program reduces the cost; \textbf{(2)} developing efficient programming schemes for local unitary gates subject to algebraic constraints, such as stabilizer gates \cite{GeometryStablizerStates2014} and locally symmetric unitaries \cite{Marvian2022}. \textbf{(3)} developing programming schemes for noisy NISQ devices, in which the program state and the processor's POVM may be significantly degraded by Markovian or non-Markovian noise \cite{singh2025informationstoragetransmissionmarkovian}. To achieve precision comparable to that of ideal, noiseless programming, explicit error correction and additional unitary queries might be required during the learning phase, thereby increasing both circuit complexity and overall program cost.

\section*{Acknowledgments}
Both the authors contributed substantially to the research presented in this paper and to the preparation of the manuscript.

E.H. thanks Mingnan Zhao and Minglong Qin for insightful discussions on unitary designs. 
This work is supported by the National Natural Science Foundation of China via the Excellent Young Scientists Fund (Hong Kong and Macau) Project 12322516, the National Natural Science Foundation of China (NSFC)/Research Grants Council (RGC) Joint Research Scheme via Project N\_HKU7107/24, the Hong Kong Research Grant Council (RGC) through the General Research Fund (GRF) grant 17303923, and the Guangdong Provincial Quantum Science Strategic Initiative via Project GDZX2403008.

\bibliographystyle{quantum}

{\small \bibliography{Quantum_Accepted/accepted_refs}}

\end{document}